\documentclass[letter,11pt]{article}
\pdfoutput=1

\usepackage[affil-it]{authblk}
\usepackage{amssymb, amsthm, amsmath}
\usepackage{hyperref}
\usepackage{braket}
\usepackage{graphicx}
\usepackage[margin=1in]{geometry}
\usepackage{algorithm}
\floatname{algorithm}{Protocol}

\usepackage{color}
\newcommand{\norm}[1]{\left|\left|#1\right|\right|}
\newtheorem{theorem}{Theorem}
\newtheorem{corollary}{Corollary}
\newtheorem{lemma}{Lemma}
\newtheorem{definition}{Definition}

\begin{document}
\title{Rigidity of quantum steering and one-sided device-independent verifiable quantum computation}

\author[1]{Alexandru Gheorghiu\footnote{a.gheorghiu@sms.ed.ac.uk}}
\author[1]{Petros Wallden}
\author[1,2]{Elham Kashefi}
\affil[1]{School of Informatics, University of Edinburgh,}
\affil[ ]{10 Crichton Street, Edinburgh EH8 9AB, UK}
\affil[2]{CNRS LTCI, Departement Informatique et Reseaux,}
\affil[ ]{Telecom ParisTech, Paris CEDEX 13, France}

\date{}
\maketitle

\begin{abstract}
The relationship between correlations and entanglement has played a major role in understanding quantum theory since the work of Einstein, Podolsky and Rosen \cite{epr}. Tsirelson proved that Bell states, shared among two parties, when measured suitably, achieve the maximum non-local correlations allowed by quantum mechanics \cite{tsirelson}. Conversely, Reichardt, Unger and Vazirani showed that observing  the maximal correlation value over a sequence of repeated measurements, implies
that the underlying quantum state is close to a tensor product of maximally entangled states and, moreover, that it is
measured according to an ideal strategy \cite{ruv2}. However, this strong rigidity result comes at a high price, requiring
a large number of entangled pairs to be tested.
In this paper, we present a significant improvement in terms of the overhead by instead considering quantum steering where the device of the one side is trusted.
We first demonstrate a robust one-sided device-independent version of self-testing, which characterises the shared state and measurement operators of two parties up to a certain bound. 
We show that this bound is optimal up to constant factors and we generalise the results for the most general attacks. This leads us to a rigidity theorem for maximal steering correlations.
As a key application we give
a one-sided device-independent protocol for verifiable delegated quantum computation, and compare it to other existing protocols, to highlight the cost of trust assumptions. Finally, we show that under reasonable assumptions, the states shared in order to run a certain type of verification protocol must be unitarily equivalent to perfect Bell states.
\end{abstract}

\vspace{2pc}

\section{Introduction}
Quantum \emph{steering correlations} first appeared in the seminal paper of Einstein, Podolsky and Rosen \cite{epr} to  support their argument that 
quantum mechanics is incomplete. It was  
later formally introduced by Schr{\"o}dinger \cite{sch}. 
The observation made, was that measurements performed on one half of a bipartite entangled state can steer the state of the other half. This means that the reduced state of one side can be correlated with the classical outcome of the other party in a way that is possible only if the two parties shared entanglement (assuming the correctness of quantum mechanics).
A similar effect occurs when examining purely classical correlations, which led Bell to derive his inequalities and reveal the non-local character of quantum mechanics \cite{bell}. 
The study of correlations that characterise quantum systems has developed much since then, and now includes the already mentioned non-local and steering correlations and quantum discord correlations \cite{discord1, discord2, discord3}. Observing these correlations, given suitable assumptions, is an indication that a particular system behaves quantum mechanically 
and is used to verify the quantumness of that system. Since these correlations are specific to quantum systems, it is also anticipated that they could be
the source of certain new practical applications. 

In particular, the existence of
\emph{non-local correlations}, apart from revealing a counter intuitive feature of nature, has led to the development of \emph{device-independent} protocols for quantum key distribution (QKD), quantum random number generation (QRNG) and verified delegated quantum computation (VDQC) \cite{ruv2, ruv, gkw, hpf, mckague, deviceindependentqkd, diqkdsecurity, Pironio2010, diqrng}.
In these applications two parties, Alice and Bob, do not trust their devices throughout the run of the protocol. Instead, by obtaining non-local correlations between the classical outputs of their devices, they can
test their devices and obtain the correct functionality.
This is achieved by confirming that the states they were sharing were entangled and measured in such a way that the resulting correlations could not have come from any classical system. Thus, the two parties 
either obtain a correct and secure result
or detect that the devices are working incorrectly (thus being insecure) and abort.
Such protocols are highly desired for practical implementation as they provide a higher level of security, unachievable by classical systems. 
However, there are certain practical issues 
that hinder their development such as the need for high detection thresholds, high fidelity transmission channels, space-like separation and a high overhead \cite{PABGMS, deviceindependentqkd}.

The practical limitations of device-independent protocols has motivated the revival of research into quantum steering, which has simpler trust assumptions, where the state of one (trusted) side is steered by operations on the other (untrusted) side. The existing research involves the characterisation of steering correlations both  
 analytically and geometrically \cite{steering1, steering2, steering3}, their relationship to other types of correlations \cite{steering4, steering5} and their application to cryptographic tasks such as QKD and QRNG \cite{BCWSW12, QRNGsteering}. Experiments testing quantum steering inequalities \cite{loophole_steering} (loophole-free) and testing local but steerable states \cite{local_steering} have also been performed.
In the case of QKD, Branciard et al. showed in \cite{BCWSW12}, that there is a natural correspondence between the trust assumptions of the protocol and the types of correlations between the two parties. Using this correspondence, Branciard et al. introduced \emph{one-sided device-independent} QKD, which uses steering correlations in order to distil a shared secret key. In the cryptographic setting, such correlations allow only one device to be untrusted leading to a reduction in the overall experimental requirements of the protocol. 
To be precise, they showed that in typical device-independent settings, the detection efficiency of Alice and Bob should be above $91.1\%$, whereas their one-sided protocol lowers that to $65.9\%$.
A similar relation between trust assumptions and correlations is exploited for QRNG as well \cite{QRNGsteering}. In this case it was shown that a detection efficiency of $50\%$ is sufficient for random number generation in the steering setting, versus $70.7\%$ in the device-independent setting.

For VDQC, we will show that using steering correlations leads to a reduction in communication, compared to the device independent setting. However, the VDQC case presents a complication.
Firstly, the setting of VDQC is slightly different than that of QKD and QRNG. In the latter two, Alice and Bob are two parties that are working together towards a common objective (obtaining a shared key or certified randomness), using possibly untrusted quantum devices. In VDQC, Alice is a client who is delegating a difficult computation to Bob, an untrusted quantum server. She does not have the resources to perform the quantum computation herself, and while Bob does, he cannot be trusted to do so. So Alice needs a way to verify that Bob is performing the correct quantum computation. To do this, she utilizes a quantum device in her local lab which she may or may not trust. Regardless of this, Bob is always assumed to be unstrusted. This is an important distinction from the collaborative setting of QKD and QRNG.
Moreover, approaches to QKD and QRNG rely on using a bound on the correlations of the parties' devices in order to derive a bound for a quantity of interest, such as key rate, mutual information, entropy etc \cite{Pironio2010, boundEnt, HHHO, boundsQKD, randomnessExpansion, Colbeck2012}. 
In contrast to this, existing protocols for device-independent VDQC use the bound on correlations to recover the underlying quantum state used in the protocol, as well as the operations being performed on this quantum state \cite{ruv, gkw, mckague, hpf}. This allows for the correctness certification (referred to here as verification) of an arbitrary universal quantum computation strictly from the non-local correlations. To achieve this one needs, at the same time, to obtain close to \emph{maximal} non-local correlations \emph{and} not only recover a characterisation of one Bell pair but of a tensor product of Bell pairs. Such a result is possible due to the rigidity of repeated CHSH games, manifesting non-local correlations, as shown in \cite{ruv}.
In this context, rigidity means that if two non-communicating devices playing CHSH games are achieving the optimal win rate, then they share a state which is close to a tensor product of maximally entangled states and, moreover, their strategy is fixed and uniquely determined. In general, we define rigidity as the ability to derive a robust bound on the distance between a real state (the state shared by Alice and Bob) and many copies of some target state (such as a tensor product of Bell pairs) as well as between target measurements and real measurements, from a bound on correlations. This is similar to the concept of \emph{self-testing}, except that for self-testing one only wants to obtain a single copy of the target state from observed correlations. For instance, it is possible to self-test a Bell state and its associated measurements using the CHSH game. So, one way to arrive at rigidity is to consider multiple CHSH games and combine the self-testing results\footnote{This needs to be done in a non-trivial way since the games might not be independent from each other, in general.}.
However, in general, it is necessary to play a large number of games in order to certify few Bell states. For this reason
the existing device-independent VDQC protocols have impractically large communication complexity \cite{ruv, gkw, mckague}. \\
There is an additional aspect to be mentioned. As stated, in the VDQC setting we have a \emph{trusted} client, also referred to as \emph{verifier}, and an \emph{untrusted} server, also referred to as \emph{prover}.
This is the standard cryptographic scenario when considering verification of computation, whether it is quantum or classical \cite{abe, cver}.
The asymmetry in trust is precisely the same as in the steering scenario which is why VDQC is the most natural application for steering correlations. While it is definitely possible to introduce such an asymmetry in QKD, for example, the typical setting is to have the parties involved be identical in all respects.

In this paper, we relax the trust assumptions and derive a rigidity result for quantum steering correlations, showing that it leads to a reduced overhead compared to the setting of non-local correlations. The result proves that observing steering correlations close to their maximal value, implies a tensor product structure of Bell pairs and fixed measurement operators for the untrusted party, up to local isometry. 
The only assumptions made in deriving this result are the correctness of quantum mechanics and the fact that one party is completely trusted, having a complete characterisation of her Hilbert space and measurement operators.
In particular, the rigidity result makes no i.i.d. (independent and identically distributed) assumption regarding the shared state or strategy of the untrusted party.
In analogy to one-sided device-independent QKD and QRNG, this leads us to a one-sided device-independent VDQC protocol, having improved round complexity over the device-independent versions. More generally, the rigidity result is relevant in its own right since it is applicable to any protocol that uses steering correlations.

The structure of this paper is organized as follows. In section~\ref{sect:steering} we give our main result, which is that using maximal steering correlations we determine,
up to a local isometry, a tensor product of Bell pairs and the operations of the untrusted party. 
To derive this result, we first give a procedure to characterise a single Bell pair from observed correlations (single-shot rigidity), in the i.i.d. setting in subsection~\ref{subsect:selftesting}. This is done by making a protocol for self-testing from steering correlations which gives us a bound on the distance of one shared state from a perfect Bell pair. We also prove the optimality of this bound. We then remove the i.i.d. assumption in subsection~\ref{subsect:non-iid}, thus showing that one can extract a Bell pair from the observed statistics even in the fully adversarial (one-sided) setting.
Then, in subsection~\ref{subsect:rigidity} we use the previous result in order to determine a tensor product structure of Bell pairs and the measurement operators of the untrusted party. 
These steps are shown in Figure~\ref{fig:steps}.
Note that rigidity does not follow directly from repeatedly applying the single-shot result as this would require independence. Instead, we use an approach similar to that of \cite{ruv}, by defining a quantum steering game and showing that high win rates in this game determine the states and strategies of the players, up to local isometry. 

\begin{figure}[htbp!]
\centering
\includegraphics[scale=0.38]{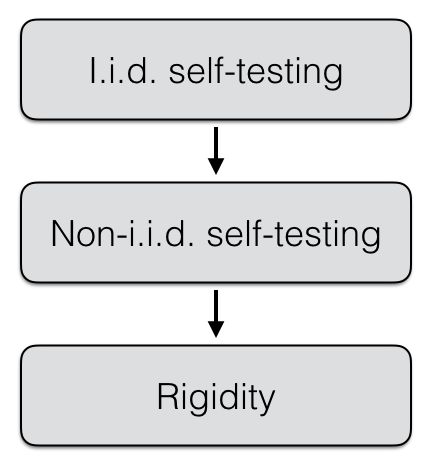}
\caption{Steps towards rigidity}
\label{fig:steps}
\end{figure}

Lastly, in section~\ref{sect:verification} we briefly discuss existing approaches to VDQC (subsection~\ref{subsect:old_VDQC}) and then use the rigidity of quantum steering  to construct  a one-sided device-independent verification protocol (subsection~\ref{subsect:new_VDQC}).  We also show that for the types of protocols we have considered, the required entangled states should be close to Bell pairs (subsection \ref{subsect:totally_steerable}).

\section{State and strategy certification via steering}\label{sect:steering}

While quantum steering has been studied extensively in the context of verifying entanglement, it is important to elaborate on the subtle difference between verifying entanglement and verifying maximal entanglement and how this relates to the verification of quantum computation. It has already been shown that it is possible to verify, from steering correlations, that a state shared by two parties is entangled. In fact, this type of verification can be done in a fully device-independent way (under certain assumptions), and has been tested experimentally \cite{untrustedsteering, disteering}. However, it should be noted that these results use 
steering correlations as a witness for quantum entanglement. The purpose of a witness is to separate between entangled and non-entangled states and its existence is proven through the violation of a steering inequality, in analogy to a Bell inequality. 
In our setting, we do not just require correlations that violate a steering inequality, rather we require the correlations to  
saturate the maximum achievable value. This enables us to certify the state and strategies used in producing these correlations, which in turn, can be used for verifying quantum computations. To our knowledge, there are no other results which treat the case of saturated steering correlations.

The setting that we will consider is similar to the one used in \cite{steering4, BCWSW12}. This involves two parties, Alice and Bob, where Alice has a trusted measurement device, while Bob has an untrusted measurement device. They share an unknown joint quantum state, $\Ket{\psi}$, which without loss of generality can be assumed to be pure. Alice instructs Bob to perform a measurement on their joint state.
For example, if the shared state is $\Ket{\phi_+} = (\Ket{00} + \Ket{11})/\sqrt{2}$, Alice can instruct Bob to measure the $Z$ observable on his qubit and report the outcome. The measurement steers Alice's qubit to a particular quantum state. She can then measure her state to confirm that her qubit was indeed steered to the expected state.
We shall refer to a collection of such measurements as a \emph{steering game} consisting of rounds of single measurements performed by Alice and Bob. They win a round of the game if Bob's reported outcome matches that of Alice. 
We elaborate more on this in subsection~\ref{subsect:rigidity}.

The correlations between their outcomes are called steering correlations if they cannot be explained by a \emph{local hidden state model}.
This happens when the expected value for the correlation of their outcomes obeys certain steering inequalities. However, whenever a steering inequality violation is observed, Alice can conclude that her state was indeed steered by Bob via their shared entanglement.
In the rest of this section we prove that violating the steering inequalities maximally (up to order $O(\epsilon)$) leads to recovering a tensor product of Bell pairs with measurement operators close to ideal and we quantify this ``closeness''.
Proving this rigidity result is done as follows.
First, building on the work regarding self-testing the singlet \cite{selftesting} by McKague, Yang and Scarani, we derive a robust self-testing protocol of a single Bell state, where one side is trusted, while the other is not. 
This is achieved by using maximal steering correlations in order to fully characterise, up to a bound of order $O(\sqrt{\epsilon})$, the 
quantum state shared between the two parties. 
Importantly, in \cite{selftesting} it is assumed that the quantum states are i.i.d. We also make this assumption in our self-testing result (subsection~\ref{subsect:selftesting}), however we remove it later on (subsection~\ref{subsect:non-iid}).
We show that the bound for self-testing
is tight up to constant factors. Then, we remove the i.i.d. assumption, arriving at a new bound for characterising the shared state in the completely adversarial setting, from the observed correlations. The way in which we remove the i.i.d. assumption is not specific to steering and can be applied to the non-local setting as well, thus complementing the work from \cite{selftesting}.
Using this result, and a game-based argument we show that saturating the steering inequalities enables us to recover the quantum state of a tensor products of Bell pairs and characterise the untrusted measurement operators acting on these states, thus proving the rigidity of quantum steering (subsection~\ref{subsect:rigidity}).
Throughout this paper we use $||\Ket{\psi}|| = \sqrt{\Braket{\psi | \psi}}$ as the $l^2$-norm and $TD(\rho, \sigma) = \frac{1}{2} Tr(\sqrt{(\rho - \sigma)(\rho - \sigma)^{\dagger}})$ as trace distance. Additionally, for the trace distance of pure states we will write $TD(\Ket{\psi}, \Ket{\phi}) = TD(\Ket{\psi}\Bra{\psi}, \Ket{\phi}\Bra{\phi})$.

\subsection{One-sided device-independent self-testing (i.i.d.)} \label{subsect:selftesting}
We start by proving a one-sided device-independent version of robust self-testing, whereby the measurement statistics allow us to determine the existence of a single Bell pair. Specifically, as in \cite{selftesting}, successfully self-testing a maximally entangled state between Alice and Bob means:
\begin{itemize}
\item There exist local bases (local isometry) in which their shared state can be viewed as a Bell pair, possibly in tensor product with some additional state.
\item We can infer the existence of local (physical) observables on Alice and Bob's side, which act non-trivially on the shared state.
\end{itemize}
This is similar to the works of \cite{hpf, selftesting, mayersyao,  selftesting2, selftesting3, selftesting4, selftesting5, PhysRevA.91.052111} and in fact we adopt a similar notation to that of \cite{selftesting}.
The main difference with respect to those works, is that in our case we trust Alice's measurements.
The specific observables we consider for her are Pauli $X$ and $Y$. We will find that Bob must also measure these observables to saturate the correlations of measurement outcomes. The Bell state under consideration is the $XY$-determined Bell state $\Ket{\psi_+} = (\Ket{01} + \Ket{10})/\sqrt{2}$.
The result can be generalized for any pair of non-commuting observables and any Bell state. 
It should be mentioned that a Bell state has a local hidden variable model for Pauli basis measurements by both parties, but it does not have a local hidden state model. This highlights the difference between non-local and steering correlations and emphasizes the importance of trusting Alice's system in order characterise the shared state and Bob's measurements.
The reason for our choice of observables and state is so that we can easily use this result for a VDQC protocol, presented in the next section, that uses $XY$-plane states, as explained in section~\ref{sect:verification}.
Note that Alice and Bob need to perform multiple measurements in order to approximate the expectation values of their observables. Since we do not trust Bob, we cannot in general assume the independence of his measurements. However, we will make an i.i.d. assumption in the beginning, to prove our main self-testing theorem. We then remove the assumption by modelling the measurement process as a martingale and using the Azuma-Hoeffding inequality \cite{azuma, hoeffding}, as is also done in \cite{Pironio2010, hpf}. A schematic illustration of our setting can be found in Figure~\ref{fig:self_test}.
\begin{figure}[htbp!]
\centering
\includegraphics[scale=0.3]{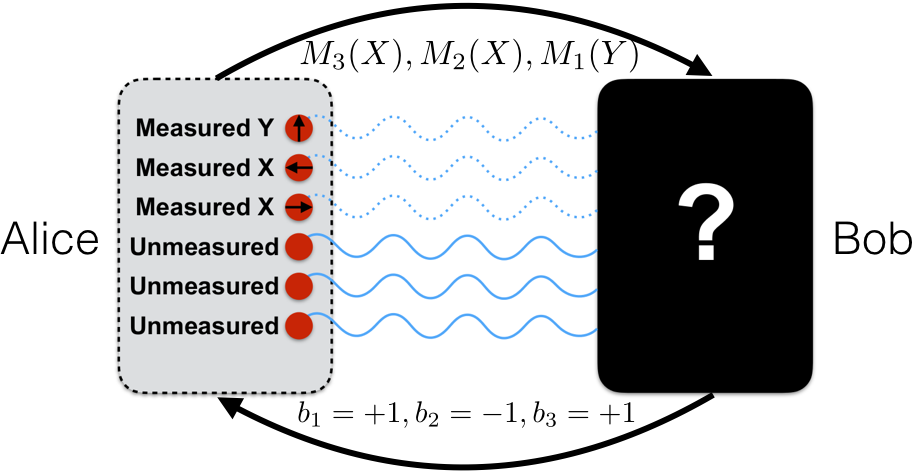}
\caption{Setting for i.i.d. self-testing}
\label{fig:self_test}
\end{figure}

We start by proving a theorem analogous
to Theorem 1 of \cite{selftesting}. Just as in \cite{selftesting}, the primed observables denote untrusted operators. The shared state, which is also assumed to be untrusted is
denoted as $\Ket{\psi}$. The independence (i.i.d.) assumption means that Alice and Bob share the same state $\Ket{\psi}$ in each round of measurements. Furthermore, the untrusted measurement observables of Bob can be assumed to be the same each time. This is because his observables can include action on his private ancilla. Since Alice's side is trusted, we
assume that she has a single qubit measurement device and so, in each round of measurement, her part of $\Ket{\psi}$
is a single qubit state.  Given this setting, the theorem is stated as follows: 
\begin{theorem}\label{theorem:selftesting}
Suppose that from the observed correlations of measurements performed by Alice and Bob and knowing that Alice is measuring the $\{X$, $Y\}$ observables (denoted $X_A$, $Y_A$), one can deduce the existence of local observables $\{X'_B$, $Y'_B\}$ on Bob's side,  with eigenvalues $\pm 1$, which act on a bipartite state $\ket{\psi}$ such that:
\begin{eqnarray}
\label{eq:th1cond1}
||(X_{A} -  X'_B) \ket{\psi}|| &\leq &  \gamma_1, \\
\label{eq:th1cond2}
||(Y_{A} -  Y'_B) \ket{\psi}|| &\leq &  \gamma_1, \\
\label{eq:th1cond3}
||(X'_B Y'_B + Y'_B X'_B) \ket{\psi}|| &\leq &  \gamma_2. 
\end{eqnarray}
Then there exists a local isometry $\Phi = I  \otimes \Phi_{B}$ and a state $\ket{junk}_{B}$ such that
\begin{equation} \label{eq:selftest}
\norm{ \Phi(M_A N'_B\ket{\psi}) - \ket{junk}_{B} M_A N_B\ket{\psi_{+}}_{AB}} \leq \varepsilon 
\end{equation}
with $M_A, N_B \in \{I, X, Y \}$, $N'_B \in \{I, X'_B, Y'_B\}$, $\varepsilon=3 \gamma_1 + \gamma_1^2/4 + 2\gamma_2$, $\ket{\psi_{+}} = (\Ket{01} + \Ket{10})/\sqrt{2}$.
\end{theorem}
\begin{proof}[Proof sketch]
\begin{figure}[htbp!]
\centering
\includegraphics[scale=1.0]{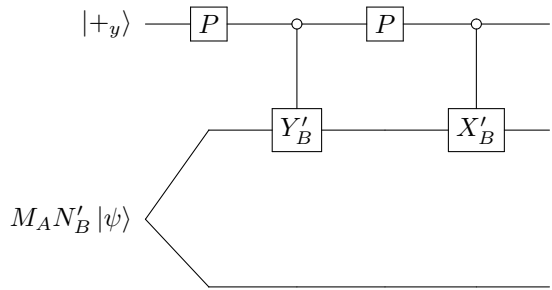}
\caption{Local isometry $\Phi$}
\label{fig:local_isometry}
\end{figure}
The proof relies on finding an isometry which, given conditions~(\ref{eq:th1cond1}-\ref{eq:th1cond3}),
maps $\Ket{\psi}$ to an almost perfect Bell state. 
Similar to \cite{selftesting}, $M_A$ and $N'_B$ are the physical observables of Alice and Bob which act on the shared state.
The isometry we considered is illustrated in Figure~\ref{fig:local_isometry}, where $P = \frac{1}{\sqrt{2}}(X + Y)$ and the control gates act on the target state when the control qubit is in the $\Ket{-_y}$ state instead of the $\Ket{1}$ state, and act as identity when the control is in the $\Ket{+_y}$ state instead of the $\Ket{0}$ state. Here, $\Ket{+_y}$ and $\Ket{-_y}$ are the two eigenstates for the Pauli $Y$ operator, corresponding to the $+1$ and $-1$ eigenvalues, respectively. The fact that we are using these states and the $P$ operator is a consequence of shifting everything to the $XY$-plane of the Bloch sphere instead of the more familiar $XZ$-plane.
It should be noted that $M_A$ is trusted and acts on Alice's part of the shared state, whereas $N'_B$, acting on Bob's part of $\Ket{\psi}$, is untrusted. However, the action of $N'_B$ is equivalent to the honest $N_B$ acting on the ancilla introduced by $\Phi$. Having the isometry, we write out its action on the state $M_A N'_B \Ket{\psi}$ and use inequalities~(\ref{eq:th1cond1}-\ref{eq:th1cond3}) together with the trace preserving properties of the operators and triangle inequalities to prove condition~\ref{eq:selftest}.
The full proof of Theorem~\ref{theorem:selftesting} can be found in Appendix~\ref{sect:selftesting}. 
\end{proof}

Our next result is to show that conditions~(\ref{eq:th1cond1}-\ref{eq:th1cond3}) are satisfied if an almost maximal violation of a particular steering inequality occurs. 
As mentioned before, the requirement for maximal violation is in contrast to previous work on entanglement detection. In that case, one uses the steering inequality as an entanglement witness to separate the Hilbert space of possible shared states into a subspace of entangled states manifesting steering correlations and its complement. Violating the inequality determines that the shared state lies in the subspace of entangled steerable states. For example, similar to the works of \cite{expSteering, cavalcanti1}, assuming Bob measures local observables $X'_B$ and $Y'_B$, one could consider the inequality:
\begin{equation*}\label{eq:steeringineq}
| \Bra{\psi} X_A X'_B + Y_A Y'_B \Ket{\psi} | \leq \sqrt{2}
\end{equation*}
This inequality holds, whenever there is a local hidden state model for Bob's system. If this is not the case, 
then the state is steerable. In our case, we do not simply require a violation of this inequality, but we require a (close to) maximal violation.
For the above inequality,
the maximum achievable violation allowed by quantum mechanics
is $2$, which also corresponds to its mathematical maximum.
Requiring maximal saturation is represented as:
\begin{equation*}
| \Bra{\psi} X_A X'_B + Y_A Y'_B \Ket{\psi}| \geq 2 - \epsilon
\end{equation*}
for some small $\epsilon$. Given that this inequality holds we have:

\begin{theorem}\label{theorem:steeringrobust}
Suppose Alice measures the observables $X_A$, $Y_A$ and that Bob measures the observables $X'_B$ and $Y'_B$ with eigenvalues $\pm 1$, on the state $\ket{\psi}$, such that
\begin{equation*}
	|\bra{\psi}\left(   X_A X'_B + Y_A Y'_B     \right)  \ket{\psi}| \geq 2 -\epsilon \label{eqn:saturate}
\end{equation*}
where $0<\epsilon<1$. Then the conditions of Theorem~\ref{theorem:selftesting} are satisfied with $\gamma_1 = \sqrt{2\epsilon}$ and $\gamma_2 = 4\sqrt{\epsilon}$.
\end{theorem}
\begin{proof}[Proof sketch]
The proof reduces to expanding inequality~\ref{eqn:saturate} and using the properties of the observables, to arrive at the bounds of conditions~(\ref{eq:th1cond1}-\ref{eq:th1cond3}) from the previous theorem. 
Concretely, we see that the correlation of local observables that we consider:
\begin{equation*}
	|\bra{\psi}\left(   X_A X'_B + Y_A Y'_B     \right)  \ket{\psi}|
\end{equation*}
is simply a sum of two expectation values which are upper bounded by unity (because the observables have $\pm 1$ eigenvalues). Hence, to saturate the absolute value of this quantity, it must be the case that both expectation values are saturated i.e. lower bounded by $1 - \epsilon$ or upper bounded by $-1 + \epsilon$. We will only examine the first case since the second is analogous, so we will drop the absolute value of the expression and simply consider:
\begin{equation*}
 \Bra{\psi} X_A X'_B + Y_A Y'_B \Ket{\psi} \geq 2 - \epsilon
\end{equation*}
By expressing each expectation as a trace norm we arrive at conditions~\ref{eq:th1cond1} and~\ref{eq:th1cond2}.
To prove condition~\ref{eq:th1cond3} we use the Cauchy-Schwarz inequality and the commutators $[X_A, Y_A]$ and $[X'_B, Y'_B]$, respectively.
The full proof can be found in Appendix~\ref{sect:steeringrobust}. 
\end{proof}

Using the results of Theorems~\ref{theorem:selftesting} and~\ref{theorem:steeringrobust}, we find that strong correlations between Alice and Bob's measurement outcomes, given that we trust Alice to be measuring the $X$ and $Y$ observables, determine the shared state between Alice and Bob as a Bell state, under local isometry $\Phi = I_{A}  \otimes \Phi_{B}$.
Additionally, notice that if the steering correlations are $O(\epsilon)$ close to ideal (maximal), we can bound the shared state of Alice and Bob as being $O(\sqrt{\epsilon})$ close to the ideal.
The same asymptotic bound is achieved in the case of CHSH games, where both Alice and Bob's outcomes are untrusted. One could expect that the bound we obtained from the steering inequalities is not a tight bound and that it is possible to do better because we trust Alice's measurements. We prove, that in fact this is not the case and the $O(\sqrt{\epsilon})$ bound is actually tight:
\begin{theorem}\label{theorem:tightbound}
Suppose that Bob's observables $X'_B$ and $Y'_B$ with eigenvalues $\pm 1$, acting on a state $\ket{\psi}$, are such that:
\begin{equation*}
	|\bra{\psi}\left(   X_A X'_B + Y_A Y'_B     \right)  \ket{\psi}| \geq 2 -\epsilon \label{eqn:saturate2}
\end{equation*}
where $0<\epsilon<1$. Then, up to constant factors, the bound of Theorem~\ref{theorem:selftesting} (i.e. inequality~\ref{eq:selftest} with $\varepsilon = O(\sqrt{\epsilon})$) is tight.
\end{theorem}
\begin{proof}[Proof sketch]
Theorem~\ref{theorem:tightbound} essentially shows that the $O(\sqrt{\epsilon})$ bound for the closeness of $\Ket{\psi}$ to the ideal Bell state $\Ket{\psi_+}$ is optimal. The proof relies on finding a state and local observables for Bob such that the steering correlation is saturated up to order $O(\epsilon)$, but the state is $O(\sqrt{\epsilon})$ deviated from the ideal Bell state. We consider such a state which is exactly $\sqrt{\epsilon}$-close to the ideal $\Ket{\psi_+}$ Bell state and which saturates the correlation of observables to $2 - \epsilon$.
Furthermore, we consider local observables for Bob which are deviated, as a function of $\epsilon$, from the ideal observables. Hence, these deviated observables tend to the ideal $X$ and $Y$ observables in the limit where $\epsilon \rightarrow 0$.
The specific state and Bob's local observables are given in Appendix~\ref{sect:tightbound}.
\end{proof}

The exact bound for closeness can be computed by simply inserting the constants $\gamma_1$ and $\gamma_2$ from Theorem~\ref{theorem:steeringrobust} in the calculation for $\varepsilon$ of Theorem~\ref{theorem:selftesting}. This yields a distance of $\varepsilon = (3\sqrt{2} + 8)\sqrt{\epsilon} + \epsilon/2$.

This result has important consequences for the application we consider in the next section, the verification of quantum computation.
It essentially imposes a restriction on how good the fidelity of the entangled state is, as resulting from the observed correlations. In VDQC we require a complete characterisation of the entangled states that are used and in particular, we require them to be close to Bell pairs. This leads to the requirement of a very tight saturation of the steering correlations.
Contrast this to the QKD or QRNG settings where less than maximally entangled states suffice \cite{boundEnt, HHHO, tripartiteQKD, qrngVacuumStates, qrngVacuum}.

It should be noted that we used the same notation for Bob's observables in both Theorem~\ref{theorem:steeringrobust} and Theorem~\ref{theorem:selftesting} since for this particular case they coincide. In general, however, we would have to construct Bob's observables in Theorem~\ref{theorem:selftesting} from the observed correlations, as is done in \cite{selftesting}.  
An important corollary to these three theorems is the following:
\begin{corollary}\label{corr:general}
The results of Theorems~\ref{theorem:selftesting},~\ref{theorem:steeringrobust} and~\ref{theorem:tightbound} hold if instead of $X$ and $Y$, Alice measures the anti-commuting single-qubit observables $A_0$ and $A_1$, having eigenvalues $\pm 1$ and Bob measures observables $B_0$, $B_1$ having eigenvalues $\pm 1$.
\end{corollary}
\begin{proof}
In the proof of Theorem~\ref{theorem:selftesting} we only made use of the anti-commutation properties of the $X$, $Y$ observables on Alice's side as well as the action of the two operators on the eigenstates of $Y$. For general observables, $A_0$ and $A_1$ this translates to using their anti-commutation properties and the action of the two on the eigenstates of $A_1$, for example.
Essentially the proof of Theorem~\ref{theorem:selftesting} only changes by relabelling $X_A$ as $A_0$ and $Y_A$ as $A_1$. On Bob's side, the situation is similar. By relabelling $X'_B$ as $B_0$ and $Y'_B$ as $B_1$ we again have conditions~(\ref{eq:th1cond1}-\ref{eq:th1cond3}) for these observables, which are then used to construct the isometry and prove the result of Theorem~\ref{theorem:selftesting}.

For Theorem~\ref{theorem:steeringrobust}, using the same relabelling we have the inequality:
\begin{equation*}
|\Bra{\psi} A_0 B_0 + A_1 B_1 \Ket{\psi}| \geq 2-\epsilon
\end{equation*}
Since the proof of Theorem~\ref{theorem:steeringrobust}, like Theorem~\ref{theorem:selftesting}, only relied on the anti-commutation properties and the action of the observables on the eigenstates of one of them, the relabelling does not change anything and the results go through as before.

Lastly, for Theorem~\ref{theorem:tightbound} we use that fact that the observables $A_0$ and $A_1$ are linear combinations of Pauli matrices. Depending on which Bell state is stabilized by the actions of these observables we can consider a state that is $O(\sqrt{\epsilon})$-deviated from that Bell state, similar to the state considered for the proof of Theorem~\ref{theorem:tightbound}.
Analogously, we will have $O(\epsilon)$-deviated observables for each Pauli matrix. For example, in the proof of Theorem~\ref{theorem:tightbound} we considered the $O(\epsilon)$-deviated versions of the $X$ and $Y$ observables. Bob's observables will therefore be linear combinations of these $O(\epsilon)$-deviated Paulis so as to be $O(\epsilon)$-close to $A_0$ and $A_1$, respectively.
\end{proof}

The result of Theorem~\ref{theorem:steeringrobust} assumes ideal expectation values for the observables of Alice and Bob. Of course, in practice, after performing a finite number of measurements we obtain an approximation of these expectation values. 
This can be properly taken into account by considering independent random variables (since we are in the i.i.d. setting) associated with the measurement process and using a Chernoff inequality to bound their expectation values. 
We do not give a full derivation of this here, since we will give the more general derivation for the non-i.i.d. case in the next section (for which the proof can be found in Appendix~\ref{sect:azuma}). Instead, we simply state the result of this finite analysis: for a fixed $\epsilon > 0$, we require at least $(1/\epsilon^2) log(1/\epsilon)$ measurements in order to certify that the closeness of each shared state is $O(\sqrt{\epsilon})$ to a perfect Bell pair.
One can also compute the number of measurements as a function of the desired trace distance for the Bell states. If we denote this distance as $D = c \sqrt{\epsilon}$, then the number of measurements must be at least $(2c^4/D^4) log(c/D)$.
In our case $c \approx 12.3$, so that if we wanted the trace distance to be, for example, $D = 0.1$, we would require at least $2.2 \times 10^{9}$ measurements.

\subsection{Removing the independence assumption} \label{subsect:non-iid}
We proceed to remove the i.i.d. assumption from the previous statements. 
The following theorem essentially states that if Alice and Bob are asked to perform a sequence of measurements, and we notice a close to maximal steering inequality violation from their outcomes, we can conclude that the state shared in a typical round of measurement is close to a Bell pair. By ``typical round'' we mean a uniformly randomly chosen round. 
A similar result is obtained in \cite{BNSVY}, with the essential differences that their non-i.i.d. result shows that at least one state is close to an ideal Bell pair while both parties are untrusted.
In our case, Alice is trusted throughout this process and, without loss of generality, we can assume that she chooses the measurement settings for each round. The notation $Tr_{-i}(\cdot)$ indicates that we are tracing out everything apart from the quantum states that are measured in round $i$. We also use the notation $Tr_{-R}(\cdot)$, which generalizes the previous notation for a set, $R$, of rounds (i.e. tracing out all states except those which are used in rounds $i \in R$).
\begin{theorem}\label{thm:azuma}
Suppose Alice and Bob are required to perform $K$ rounds of measurement and also that:
\begin{itemize}
\item Prior to the measurements, the shared state of Alice and Bob is assumed to be $\sigma$, which can be either pure or mixed. The state $\sigma$ is thus the global state which will be used for all $K$ rounds\footnote{In the ideal setting where everything is trusted, $\sigma$ would be a $2K$-qubit state consisting of $K$ Bell pairs.}.
\item Alice chooses a random set of size $K/2$, consisting of distinct indices from $1$ to $K$ and denoted $R_0 = \{ i | i \in_{R} \{1 ... K\} \}$, $|R_0| = K/2$. We also denote $R_1 = \{1 ... K \} \backslash R_0$, to be the complement of $R_0$.
\item We denote $\rho_i = Tr_{-i}(\mathcal{E}^{AB}_{1,i-1}(\sigma))$ the reduced state of Alice and Bob in round $i$, and:
\begin{equation*}
\rho_{avg} =
\frac{1}{K} \sum\limits_{i = 1}^{K} \rho_i
\end{equation*}
 as the averaged state. Here $\mathcal{E}^{AB}_{1,i-1}$ denotes the action (measurements) of Alice and Bob on the state $sigma$ up to round $i$.
\item In round $i$, let $r_i = 0$ iff $i \in R_0$, otherwise $r_i = 1$. Alice measures the observable $A_{r_i}$ on her half of $\rho_i$. $A_0$ and $A_1$ are anti-commuting single-qubit observables having $\pm 1$ eigenvalues.
\item In round $i$, Bob is asked to measure $B_{r_i}$. $B_0$ and $B_1$ have $\pm 1$ eigenvalues.
\item We denote $a_i$ and $b_i$, respectively, as the outcomes of their measurements in round $i$. We also denote $\hat{C}_i = a_i b_i$ as their correlation for round $i$.
\item We denote $\hat{C}^0 = \frac{1}{K/2}\sum\limits_{i \in R_0} \hat{C}_i$ and $\hat{C}^1 = \frac{1}{K/2}\sum\limits_{i \in R_1} \hat{C}_i$ as the averaged correlations for the cases where both Alice and Bob are asked to measure the first observable, or both are asked to measure the second, respectively.
\end{itemize}
If, for some given $\epsilon > 0$ and suitably chosen $K = \Omega((1/\epsilon^2) log(1/\epsilon))$, it is the case that $ \hat{C}^0  +  \hat{C}^1  \geq 2 - \epsilon$ (or, alternatively, $ \hat{C}^0  +  \hat{C}^1  \leq -2 + \epsilon$)  then there exists a local isometry $\Phi$ such that, for a randomly chosen $\rho_i$, with probability at least $1 - O(\epsilon^{1/6})$:
\begin{equation} \label{eq:reduced_state}
TD( Tr_{junk}(\Phi( \mathcal{E}^{AB}(\rho_{i}) )),  \hat{\mathcal{E}}^{AB}(\Ket{\psi_+}\Bra{\psi_+})) \leq O(\epsilon^{1/6})
\end{equation}
Where $\mathcal{E}^{AB}$ is some combination of the $A_0, A_1, B_0, B_1$ operators and $\hat{\mathcal{E}}^{AB}$ is the analogous combination of the ideal operators (i.e. $A_0$, $A_1$, $I$ and the ideal operators for Bob, which for the $XY$-plane case we considered, are $I, X, Y$), as in Theorem~\ref{theorem:selftesting}, and where $junk$ is Bob's private system apart from the ancilla introduced by $\Phi$. Alternatively, we have that there exists some state $\tilde{\rho}_{junk}$ such that:
\begin{equation} \label{eq:reduced_state2}
TD( \Phi(\mathcal{E}^{AB}(\rho_{i})),  \hat{\mathcal{E}}^{AB}(\Ket{\psi_+}\Bra{\psi_+}) \tilde{\rho}_{junk}) \leq O(\epsilon^{1/12})
\end{equation}
\end{theorem}
\begin{proof}[Proof sketch]
The proof is broken down into several parts. First, we show that the average observed correlations $\hat{C}^0$ and $\hat{C}^1$ approximates the ideal quantum correlation for the averaged state. The averaged state can be thought of as the state shared by Alice and Bob in each round of measurements, such that the average correlations of outcomes from this state match those observed in the real experiment (i.e. $\hat{C}^0$ and $\hat{C}^1$).
Proving this step is done along similar lines to the approaches of \cite{Pironio2010, hpf}. The measurement process of Alice and Bob can be viewed as a stochastic process with bounded increment, i.e. a martingale. The specific martingale we consider encodes the correlations of their measurement outcomes. While the individual measurements need not be independent, we can still prove that this observed correlation is, with high probability, close to the ideal quantum correlation. To do this, we use the Azuma-Hoeffding inequality for martingales \cite{azuma, hoeffding}.
We then use the result of Theorem~\ref{theorem:steeringrobust} to show the closeness of the averaged state to an ideal Bell state.
Lastly, we prove Equation~\ref{eq:reduced_state} by using an optimization argument together with properties of trace distance and density matrices. Equation~\ref{eq:reduced_state2} follows from this through an application of the Gentle Measurement Lemma \cite{ruv}. The full proof is given in Appendix~\ref{sect:azuma}.
\end{proof}

As in the i.i.d case, for a fixed $\epsilon > 0$, we require $\Omega((1/\epsilon^2) log(1/\epsilon))$ measurements however the closeness of a typical state, to a perfect Bell pair, is of order $O(\epsilon^{1/6})$ in this case. 
For a better comparison we will also compute the number of measurements as a function of the desired trace distance for a typical state. The proofs in Appendix~\ref{sect:azuma} show that the exact number of measurements required is $(8/\epsilon^2) log(1/\epsilon)$, yielding a distance $D = c^{1/3}\epsilon^{1/6}$, where $c$ is the constant from the i.i.d. bound.
Thus, the number of measurements must be at least $(8c^4/D^{12}) log(c^2/D^6)$.
For our case, where $c \approx 12.3$ if we again take $D = 0.1$, we would require at least $3.4 \times 10^{18}$ measurements.
While this is too great for current experimental applications, the current bounds are most likely not tight and can be improved. In fact, better numeric bounds have been obtained for the i.i.d. setting, as we explain in Subsection~\ref{subsect:comparison}. To give an example, the result of \cite{Hoban16}, obtains a bound where $c = 1.19$ leading to $2.2 \times 10^{14}$ measurements.
We show a comparative plot for the number of measurements required in the i.i.d. setting, versus in the non-i.i.d. setting in Figure~\ref{fig:plot}. The graphs are represented as functions of the desired (decreasing) trace distance.
\begin{figure}[htbp!]
\centering
\includegraphics[scale=0.242]{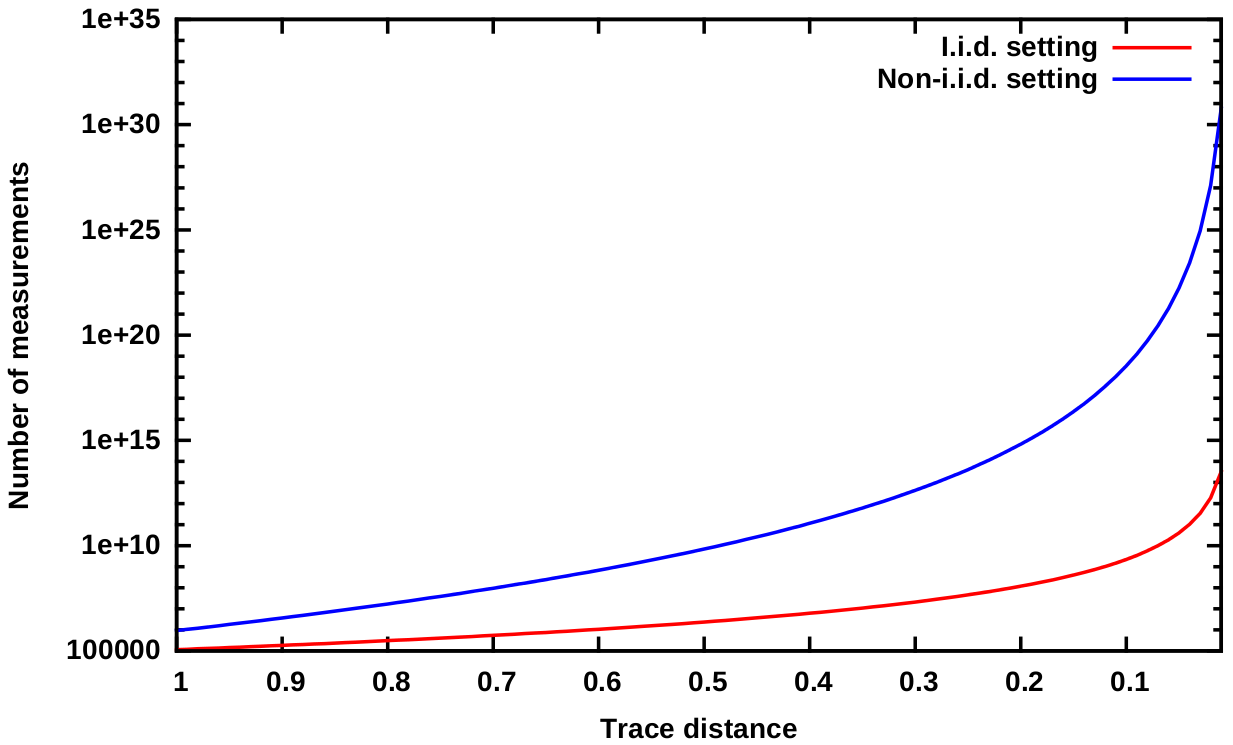}
\caption{Comparison of i.i.d. and non-i.i.d. settings}
\label{fig:plot}
\end{figure}

It should be noted that in the proof of Theorem~\ref{thm:azuma}, we did not use the fact that Alice is trusted except when applying the self-testing result (Theorem~\ref{theorem:steeringrobust}). Thus, a similar theorem can be proven in the case where both Alice and Bob are untrusted. In that case, one could simply use the self-testing results of \cite{mckague, selftesting, selftesting5, PhysRevA.91.052111} for the i.i.d. setting, and then obtain 
a statement about the closeness of a typical state to the ideal one in the non-i.i.d. setting using our techniques. For example, if we were to use Theorem 2 from \cite{selftesting} we could once again establish from the measurement statistics that a typical state shared by Alice and Bob is close to a Bell state. This result
completes the work of \cite{selftesting} for the non-i.i.d. setting.

It should additionally be noted that throughout this section we not only assumed that Alice's device is trusted but that it also measures the ideal Pauli operators. This could seem unreasonable from an experimental perspective, however note that any (fixed) deviation on Alice's measurement operator can be incorporated into $\epsilon$. In other words, assume Alice's ideal operator is $A$ and the deviated one is $\delta A$, such that:
\begin{equation*}
TD(A \otimes B \Ket{\psi}, \delta A \otimes B \Ket{\psi}) < \delta
\end{equation*}
It is thus the case that the action of Alice's operators is $\delta$-close to the action of the ideal operators which produce $\epsilon$ saturation. Hence, $\delta$ can be added to $\epsilon$ and viewed as a contribution to the total variation from maximal correlations. However, if such a deviation exists we should consider what happens when $\delta \leq\epsilon$ and when $\delta > \epsilon$, respectively. If $\delta \leq \epsilon$, then the error on Alice's device is smaller than the precision to which we wish to estimate the saturation of the correlations. Therefore, the saturation can still be considered of order $O(\epsilon)$ and the bounds on the states follow as in the ideal case. However, if $\delta > \epsilon$ then the saturation cannot be estimated within the desired precision. This means that there will be an intrinsic limitation on the determined closeness of the shared states as given by a saturation of order $O(\delta)$.

\subsection{Rigidity of quantum steering} \label{subsect:rigidity}
We now proceed to prove rigidity of quantum steering games in a manner similar to that of \cite{ruv}. 
In this setting, we assume that there is an unknown quantum state shared between Alice and Bob. By asking Alice and Bob to perform repeated measurements we would like to show that this state is close to a tensor product of Bell pairs possibly in tensor product with some additional state, up to local isometry. Additionally, we will show that Bob is essentially performing the correct measurements (recall that we are assuming Alice's measurements are trusted).
Self-testing allows us to certify one Bell state and the local observables of Bob in a one-shot manner.
Intuitively, it seems that we could perform sequential self-tests in order to certify multiple Bell pairs, thus recovering the tensor product structure. For example, according to Theorem~\ref{thm:azuma} we can conclude that after $K$ rounds of measurement
the reduced state in one of the rounds is, with high probability, close to a Bell state. We would then simply repeat this procedure $N$ times in order to recover $N$ Bell pairs, each one selected at random from each set of $K$ rounds.
This intuitive argument does not hold. As is also noted in \cite{ruv}, Bob's strategy and the untrusted states he prepares for a certain set of measurements can overlap with the states and strategy from another set.
Thus, if the reduced state for round $i$ is close to a Bell pair, and the reduced state for round $j$ is close to a Bell pair, we cannot implicitly conclude that the reduced for both rounds $i$ and $j$ is close to two Bell pairs.

We resolve this, in a manner similar to \cite{ruv}, by first defining a steering game akin to the CHSH game. 
We will then use Theorem~\ref{thm:azuma} to prove that the real strategy Alice and Bob use to play the steering game is close to the ideal strategy.
Intuitively, the steering game is one in which we ask Alice and Bob to measure specific observables on their shared state and check to see if their outcomes match. For example, to saturate inequality~\ref{eqn:saturate} we would ask Alice and Bob to both measure the $X$ observable or to both measure the $Y$ observable. 
Since we trust Alice, we know that she is behaving honestly, whereas Bob could deviate from the ideal strategy. 
Note that we are assuming that there is a referee asking Alice and Bob to perform these measurements. Since Alice is trusted, we could have her act as the referee and send the instructions to Bob. Indeed that will be the case when we consider the verification setting. However, the two situations are equivalent. For this reason, in our formal treatment of the game we shall consider Alice and Bob to be the players and that there exists a referee which instructs them on what to do. The definition is as follows:
\begin{definition}
We say that a game consisting of players Alice and Bob is a $K$-round steering game with threshold $T \leq K$ iff the following conditions are satisfied:
\begin{itemize}
\item Alice and Bob share a joint unknown quantum state $\Ket{\psi}$.
\item The game has $K$ rounds.
\item In round $i$, Alice is instructed by the referee to measure her half of $\Ket{\psi}$ with either the $A_0$ or $A_1$ two-outcome, single-qubit, anti-commuting observables, having $\pm 1$ eigenvalues and record her measurement outcome
(keeping it secret from Bob).
\item Alice's measurement device is fully trusted to perform the correct measurement, moreover she has a complete characterisation of the device's Hilbert space.
\item In round $i$, Bob is instructed by the referee to measure his half of $\Ket{\psi}$ with either the $B_0$ or $B_1$ two-outcome observables, having $\pm 1$ eigenvalues and reports his outcome to Alice.
\item Alice and Bob win the current round iff their outcomes are identical.
\item Alice and Bob win the game iff they win at least $T$ rounds.
\end{itemize}
\end{definition}
Note that in the previous definition we match the conditions of Theorem~\ref{thm:azuma}. Unlike Alice, Bob is untrusted and so his observables are unknown to Alice (and the referee). Moreover, the state $\Ket{\psi}$ is unknown and we can assume that it was prepared by Bob.
We now define the correlation value of the game.
\begin{definition}
Let $W$ be the number of rounds that Alice and Bob win in a $K$-round steering game. The correlation value for the game is defined as the fraction $W/K$.
\end{definition}

It is useful to make the following observation: if we assume that Alice and Bob are measuring the same state, $\Ket{\phi}$, in each round, then the correlation value of the game would be:
\begin{equation*}
\frac{1}{2}\Bra{\phi} A_0 B_0 + A_1 B_1 \Ket{\phi}
\end{equation*}
In general, this might not be the case, since Bob is free to use any state in each round. In accordance with Theorem~\ref{thm:azuma}, the correlation value of the game is then an estimate for the correlation of the averaged state. We use this fact to conclude something about the correlation of the reduced state in a randomly chosen round. This fact will be used to prove rigidity.

Alice and Bob will be asked to play multiple steering games. Following the notation of \cite{ruv} we will denote $\rho$
as their shared state for all these games (unlike $\Ket{\psi}$ which is the state for one game) and for a specific game, $j$, we denote the operator associated with Alice's actions (measurements) as $\mathcal{E}_j^A$ and the operator associated with Bob's action as $\mathcal{E}_j^B$.
Thus, for $N$ steering games, the triplet $\mathcal{S} = (\rho, \{ \mathcal{E}_j^A \}, \{ \mathcal{E}_j^B \})$ encodes the \emph{strategy} of Alice and Bob, where $j \leq N$.
The ideal strategy, which we denote as $\mathcal{S}_{id}$, is the one in which $\rho_{id}$ is a tensor product of Bell states and Alice and Bob perform the measurements they are instructed to perform in each round. We need to show that $\mathcal{S} \approx \mathcal{S}_{id}$.
This was shown in \cite{ruv} for CHSH games by extracting a tensor structure in the individual Hilbert spaces of Alice and Bob from the tensor structure of their two spaces. In our case, because we trust Alice and her device, we already have a characterisation of her strategy and Hilbert space. Therefore, we need only use this to characterise Bob's strategy. As in \cite{ruv}, we do this by exploiting the following symmetry property of the Bell state $(M \otimes I) \Ket{\psi_+} = I \otimes (XM^TX) \Ket{\psi_+}$, for any $2 \times 2$ matrix $M$. This essentially allows us to shift Bob's measurements to Alice's side, thus eliminating any dependence of his outcomes on previous qubits and establishing a tensor product structure.
We will use this to create a strategy in which only Alice performs measurements.
In analogy to \cite{ruv}, we start by defining an $\epsilon$-structured steering game.

\begin{definition}
We say that a $K$-round steering game is \emph{$\epsilon$-structured} iff the game correlation value is greater than $1 - \epsilon$.
\end{definition}
\noindent Alternatively, we can set the winning threshold for the game to $T = (1 - \epsilon) K$. Thus, winning a steering game is equivalent to having an $\epsilon$-structured steering game. Similarly, one can define an $\epsilon$-structured strategy.
\begin{definition}
A general strategy $\mathcal{S}$ for $N$ $K$-round steering games is \emph{$\epsilon$-structured} iff for an arbitrary game $j \leq N$ we have that $Pr(\text{game } j \text{ is }\epsilon \text{-structured}) \geq 1 - \epsilon$.
\end{definition}
\noindent Note that this means that the majority of steering games, in a strategy, are $\epsilon-$structured. Lastly, we adapt a definition from \cite{ruv} to characterize the closeness of two strategies.
It should be noted that we will define two strategies as being close if the reduced actions of Alice and Bob on $N$ randomly chosen rounds, in the two strategies, are close to each other. 
This contrasts the definition from \cite{ruv}, where two strategies are considered close if the actions of Alice and Bob for all rounds are close to each other.
The reason for the difference is that we are only interested in establishing a tensor product of $N$ Bell states and that Alice and Bob are performing ideal measurements on those $N$ states. However, since we have $N$ steering games, each consisting of $K$ rounds, we will have $NK$ rounds in total. Showing that all of these are close to ideal would lead to a large overhead which we would like to avoid.
\begin{definition} \label{def:sim}
Let $\mathcal{S} = (\rho, \{ \mathcal{E}_j^A \}, \{ \mathcal{E}_j^B \})$ and $\mathcal{\bar{S}} = (\bar{\rho}, \{ \bar{\mathcal{E}}_j^A \}, \{ \bar{\mathcal{E}}_j^B \})$ be two strategies for playing $N$ sequential $K$-round steering games. For $\epsilon \geq 0$, we say that
strategy $\mathcal{\bar{S}}$ $\epsilon$-simulates strategy $\mathcal{S}$ iff they both use the same Hilbert spaces and for all $j$:
\begin{equation*}
TD(Tr_{-R}(\mathcal{E}^{AB}_{1,j}(\rho)),  Tr_{-R}(\mathcal{\bar{E}}^{AB}_{1,j}(\bar{\rho}))) \leq \epsilon
\end{equation*}
Where $R$ is a set of $N$ round indices, each chosen at random from the $N$ steering games. Additionally:
\begin{equation*}
\mathcal{E}_{1,j}^{AB} = \mathcal{E}_1^{AB} \circ \mathcal{E}_2^{AB} \circ ... \mathcal{E}_j^{AB} 
\end{equation*}
\begin{equation*}
\bar{\mathcal{E}}_{1,j}^{AB} = \bar{\mathcal{E}}_1^{AB} \circ \bar{\mathcal{E}}_2^{AB} \circ ... \bar{\mathcal{E}}_j^{AB} 
\end{equation*}
And:
\begin{equation*}
\mathcal{E}^{AB}_{j} = \mathcal{E}^{A}_{j} \circ \mathcal{E}^{B}_{j} \; \; \; \; \; \; \; \mathcal{\bar{E}}^{AB}_{j} = \mathcal{\bar{E}}^{A}_{j} \circ \mathcal{\bar{E}}^{B}_{j}
\end{equation*}
\end{definition}

\noindent Whenever we have that $\mathcal{S}$ $\epsilon$-simulates $\mathcal{\bar{S}}$, or an isometric extension of $\mathcal{\bar{S}}$, we will write $\mathcal{S} \approx \mathcal{\bar{S}}$. This leads us to the main result:

\begin{theorem}\label{theorem:rigidity}
Let $\mathcal{S} = (\rho, \{ \mathcal{E}_j^A \}, \{ \mathcal{E}_j^B \})$ be Alice and Bob's $\epsilon$-structured strategy for playing $N$ sequential $K$-round steering games. 
Let $\mathcal{S}_{id} = (\rho_{id}, \{ \mathcal{E}_{id \; j}^A \}, \{ \mathcal{E}_{id \; j}^B \})$ be Alice and Bob's ideal strategy for playing $N$ sequential steering games.
We have that $\mathcal{S}$ $O(N\epsilon^{1/6})$-simulates an isometric extension of $\mathcal{S}_{id}$.
\end{theorem}
\begin{proof}[Proof sketch]
Firstly, it should be noted that for all $j$, $\mathcal{E}_{j}^A
=\mathcal{E}_{id \; j}^A$. This is because Alice is trusted and always playing according to the ideal strategy.
The proof then consists of two steps. First we show that if the real strategy $\mathcal{S}$ is $\epsilon$-structured, then $\mathcal{S} \approx \mathcal{S}_{g}$, where $\mathcal{S}_{g}$ is a strategy in which Alice plays honestly and also guesses Bob's measurement outcomes. The guesses provided by Alice are taken to be Bob's outcomes for each game while the action on his subsystem is taken to be identity.
The proof of this step relies on characterising the evolution of the quantum state $\rho$ in the two strategies and using Theorem~\ref{thm:azuma} together with the $\epsilon$-structured nature of the strategies.
We then show that $\mathcal{S}_{g} \approx \mathcal{S}_{id}$. 
To do this, note that in the guessing strategy we have effectively removed the problem of adaptivity. Since there is no untrusted Bob in $\mathcal{S}_{g}$, the original argument, of sequentially repeating self-testing, goes through. This allows us to show that $\mathcal{S}_{g}$ is close to a similar guessing strategy, denoted $\hat{\mathcal{S}}_{g}$, which uses ideal Bell pairs. 
Lastly, this strategy is trivially close to the ideal one.
We can then combine these results to show that $\mathcal{S} \approx \mathcal{S}_{id}$.
The full proofs are given in Appendix~\ref{sect:rigidity}.
\end{proof}

Since $N$ represents the number of Bell pairs we wish to certify, we see that in order to obtain a decreasing error, we require $\epsilon = O(N^{-6})$. We also know from Theorem~\ref{thm:azuma} that given $\epsilon$, the number of games required is of order $K = \Omega((1/\epsilon^2)log(1/\epsilon))$. Therefore, we must have that $K = \Omega(N^{12} log(N))$.
Since each steering game comprises of $K$ rounds, we have $K N$ rounds in total, or $\Omega(N^{13} log(N))$ rounds of steering games.
This becomes important for use in VDQC.

\subsection{Comparison with other approaches} \label{subsect:comparison}
Before presenting the VDQC application, we briefly compare our approach to similar results.
As mentioned, this paper builds on the work of self-testing the singlet of McKague et al \cite{selftesting}. While we use similar techniques to theirs, we assume that Alice is trusted and thus arrives at an improved bound for the closeness of $\Ket{\psi}$ to an ideal Bell pair ($O(\sqrt{\epsilon})$ vs $O(\epsilon^{1/4})$).
On the other hand, the results of \cite{ruv} and \cite{BNSVY} do arrive at an $O(\sqrt{\epsilon})$ bound for self-testing. However, we know that their techniques cannot improve the asymptotic closeness, since we have shown in Theorem~\ref{theorem:tightbound} that this bound is optimal up to constant factors.

When comparing the exact bounds, we obtained $(3\sqrt{2}+8)\sqrt{\epsilon} + \epsilon/2$, which is smaller compared to that of \cite{ruv}, approximately $270 \sqrt{\epsilon}$. The result of \cite{BNSVY} obtained numerically an even smaller factor of $\sqrt{2.2\epsilon}$  by using a semidefinite program. Their technique could, in principle, be used to improve our approach as well.
We also mention the result of {\v{S}}upi{\'c} and Hoban \cite{Hoban16}, which appeared concurrently with our own. They also consider the case of self-testing from steering correlations, obtaining an analytic bound of $13\sqrt{\epsilon}$ and a numeric bound of $1.19\sqrt{\epsilon}$.

Furthermore, \cite{BNSVY} also considers removing the i.i.d. assumption. Their approach is based on hypothesis testing, however the end result is to show that \emph{at least} one state, out of all measured states, was close to a perfect Bell pair. In our case, Theorem~\ref{thm:azuma} establishes that \emph{a typical} state, out of all measured states, was close to a perfect Bell pair.
This is necessary in order to prove the rigidity result and certify a tensor product of Bell pairs.

For the rigidity of steering correlations we employed similar techniques to those of \cite{ruv} to show that the strategy (consisting of states and measurements) associated with the real scenario is close to that of the ideal scenario. This is done by considering intermediate strategies and showing that they are close to both the real and the ideal one and therefore that the latter strategies are close to each other. The major difference with \cite{ruv} is that because Alice is trusted, in our case, there was no added overhead in proving the closeness of these strategies to each other. This then lead to a reduced closeness bound.

\section{Verified delegated quantum computation}\label{sect:verification}

The idea of VDQC is that a computationally weak verifier wants to delegate a computation to a powerful (quantum) prover, and at the same time be able to verify the correctness of the result received. In characterising VDQC protocols, we use the formalism of \emph{interactive proof systems} \cite{ip, abe}.
For our setting,
the prover is restricted to polynomial time quantum computations. 
Ideally we would like the verifier to be a fully classical computer, however it is still an open problem if this is possible, when there is
a single prover \cite{openproblem}. Instead, it is known that if the verifier has some minimal quantum capabilities, he is able to verify the prover's computation \cite{abe, fk}. 
Alternatively, if there are multiple non-communicating provers sharing entanglement it is possible to do verification with a fully classical verifier \cite{ruv, gkw, mckague}. 
It should be noted that in certain cases not all the provers are quantum computers. For example, in the protocols of \cite{gkw, hpf} there are two provers, one is a quantum computer (or server) and the other is an untrusted measurement device.
Thus, throughout this paper we will refer explicitly to the provers by their role (i.e. quantum server, measurement device etc). \\

\subsection{Existing approaches to VDQC}\label{subsect:old_VDQC}
The first approaches to VDQC relied on having a classical verifier with a minimal quantum device. 
This device could be either a constant size quantum computer \cite{abe}, or a single qubit preparation device \cite{fk}.
In both cases, the verifier has both a classical and a quantum communication channel with the server. The quantum channel is used only in an initial phase to send quantum states to the server. Verification of the computation is then performed via classical communication only. Importantly, the quantum communication is \emph{offline}, meaning it can be performed before the verifier even decides what computation she would like to perform (she must, however, fix the size of the computation to be performed).
We will briefly explain one of these protocols, namely the one by Fitzsimons and Kashefi \cite{fk}, which we shall refer to as the FK protocol. The reason for choosing this protocol is that many other VDQC protocols that were later developed have been based on FK \cite{gkw, hpf, KKD14, Morimae2014, efk, kapourniotis2015optimising, kw}.
Furthermore the FK protocol is currently the optimal protocol from the point of view of the client's quantum capability requirements. Hence it is a good starting point in further reducing the trust assumptions as we intend to do in this section.
We therefore use our results on quantum steering to modify one of the FK-based protocols.

The FK protocol is expressed in the Measurement-Based Quantum Computing (MBQC) model of computation and uses Universal Blind Quantum Computation (UBQC) \cite{bfk} as a basis for verification. We succinctly explain the basic ideas behind these concepts as further details can be found here \cite{mbqc, bfk}.
In MBQC, computation is achieved through a sequence of adaptive single-qubit measurements performed on a highly entangled state known as a \emph{graph state}. It is possible to make this graph state highly regular, by having all qubits prepared in the $\Ket{+} = \frac{1}{\sqrt{2}}(\Ket{0} + \Ket{1})$ state and entangling them using the controlled-$Z$ operation into a \emph{brickwork structure} \cite{bfk, fk}. The qubits are then measured in the basis $\{ \Ket{+_{\theta}}, \Ket{-_{\theta}} \}$, where $\Ket{\pm_{\theta}} = \frac{1}{\sqrt{2}}(\Ket{0} \pm e^{i \theta} \Ket{1})$ and $\theta$ is chosen adaptively from the set $D = \{0, \pi/4, ..., 7\pi/4 \}$.

In the following we give the main
idea behind UBQC. A trusted client sends to an untrusted quantum server rotated qubits of the form $\Ket{+_{\theta}}$ with angles $\theta$ chosen randomly from the set $D$. The server is then supposed to entangle these qubits in a generic graph state structure and then measure them in the basis $\{ \Ket{+_{\delta}}, \Ket{-_{\delta}} \}$, $\delta \in D$, as instructed by the client. The measurement angles are selected and adapted in order to perform a specific computation. Having no knowledge of the initial rotation angles (the $\theta$'s), the measurement angles (the $\delta$'s) will appear random to the server and so he will have no information about the computation being performed, apart from an upper bound on its size (given by the number of qubits).
We illustrate this in Figure~\ref{fig:ubqc}.
This blind computation procedure can be modified in order to perform verification as well, leading to the FK protocol.
\begin{figure}[htbp!]
\centering
\includegraphics[scale=0.33]{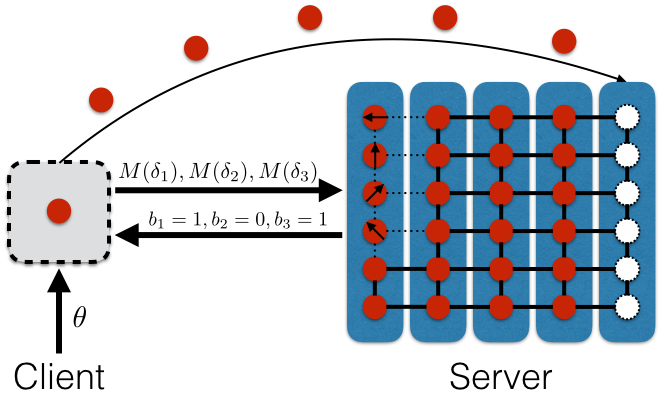}
\caption{UBQC protocol}
\label{fig:ubqc}
\end{figure}

In this case, the client, now known as a verifier, will also send computational basis states $\{ \Ket{0}, \Ket{1} \}$ to the server, interleaved randomly with the rotated qubits. The purpose of these states, called dummy qubits, is to isolate certain rotated qubits from the rest of the graph state qubits.
Isolation is achieved because the controlled-$Z$ operation when used with a dummy and a rotated qubit will not perform entanglement.
The isolated qubits are called \emph{traps} because the verifier will instruct the server to measure these qubits in their preparation basis (i.e. the measurement angle $\delta$ will match the rotation angle $\theta$ for each trap qubit), thus yielding a deterministic outcome. Because of blindness, the position of the traps is hidden from the server and so he is unaware if he is performing a trap measurement or a computation measurement. This allows the verifier to test, on average, the behaviour of the server and abort the protocol if he behaves maliciously. Schematically, this is shown in Figure~\ref{fig:verification}.
\begin{figure}[htbp!]
\centering
\includegraphics[scale=0.33]{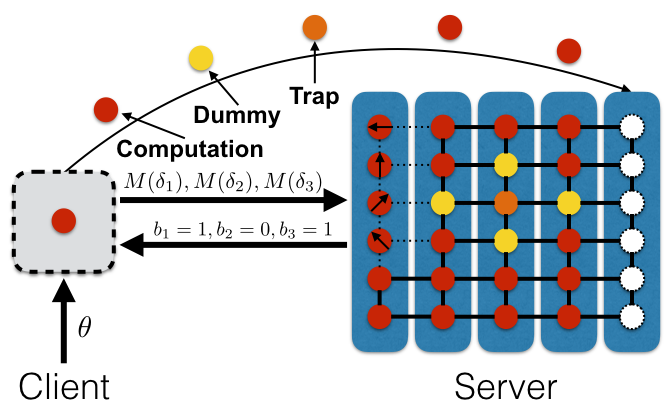}
\caption{FK verification protocol}
\label{fig:verification}
\end{figure}

As mentioned, the FK protocol uses a verifier with a trusted preparation device, while the untrusted server is performing entangling and measuring operations. We refer to this as a \emph{prepare and measure} type of protocol.
There is another class of protocols that
have the server prepare the graph state and send the qubits one by one to the verifier who performs single qubit measurements
\cite{Fujii_Morimae2012, Morimae2014, Hayashi_Morimae2015}. In this case, the verifier has a trusted measurement device instead of a preparation device. 
We refer to the latter class of protocols as \emph{measurement-only} type of protocols.
The main drawback
of this approach is that the quantum communication is online (occurs during the time that the quantum computation is performed).
This means that the verifier must commit to the computation she wishes to perform when the quantum communication commences.
In contrast to this, the FK protocol is offline and therefore the verifier does not need to make such a commitment. Instead, after the quantum communication has occurred, she is free to choose any computation up to a fixed size and communicate only classically with the server.

All of the approaches mentioned so far, relied on having a verifier with a trusted quantum device. It is also possible to have a fully classical verifier, with no quantum device, if we allow for multiple entangled quantum servers. Such is the case with the protocols presented in \cite{ruv, mckague}. Alternatively, the verifier can have an untrusted measurement device and share entanglement with a single quantum server. This scenario is a \emph{device-independent} version of the the prepare and measure protocols.
The two existing approaches are presented in \cite{gkw, hpf}. We will briefly describe the operation of these device-independent protocols, since a slight modification of them will lead to our \emph{one-sided device-independent} protocol.

The setting of device-independent VDQC is the following:
\begin{itemize}
\item The verifier has an untrusted measurement device. The device can measure the observables $\{ X', Y', Z', D', E', F' \}$, where
$D' = \frac{1}{\sqrt{2}}(X' + Z')$, $E = \frac{1}{\sqrt{2}}(X' + Y')$, $F = \frac{1}{\sqrt{2}}(Y' + Z')$. The observables are primed to indicate they are untrusted.
\item The server is instructed to initially prepare Bell states and send half of each state to the verifier.
\item The measurement device and the server are assumed to be non-communicating.
\item The verifier interacts only classically with both of these devices.
\end{itemize}
The verification protocol consists of the following two stages:
\begin{enumerate}
\item \textbf{Verified state preparation} - In this stage the verifier will use the shared entangled states with the server in order to prepare single qubit states on the server's side. These are the states which, in the FK protocol, were sent to the server having been prepared by the trusted quantum device and which will later be used to perform verified computation. In this case, because the measurement device is also untrusted, the verifier will have to interact with the two devices in order to certify the preparation of correct states on the server's side. In the protocol of \cite{gkw}, this is done using the rigidity property of CHSH games, whereas the protocol of \cite{hpf} achieves this using a modified version of the Mayers-Yao self-test.
\item \textbf{Verified computation} - Having prepared the rotated states and dummy qubits on the server's side, the verifier proceeds to run the FK protocol as if she had sent the qubits to the server. It can be shown that the server is still blind following the preparation stage. Since this stage results in the preparation of imperfect states ($\epsilon$-close to the ideal), it was necessary to show that the FK protocol is robust to deviations in the initial state.
\end{enumerate}
It is worth pointing-out here, that any deviation of the server on the correctly prepared input of the FK, is detected by the verification mechanism of the FK protocol. In other words, the two stages can be separated, and the verified preparation can happen any time earlier than the verified computation.
The main disadvantage of the device-independent protocols is the large round complexity in the state preparation stage, leading to an overall large complexity compared to other approaches. 
The main reason behind this blow-up, 
is the number of measurements that need to be performed in order to certify correct state preparation and additionally the measurements that need to be performed to test the honesty of the two devices. A major factor for this increase, is the need to certify the tensor product structure of Bell states shared, before the verifier uses them to prepare the states for the FK on the server's side.
By allowing the verifier's device to be trusted, we can eliminate this last set of measurements and gain an advantage in terms of round complexity, leading to a one-sided device-independent VDQC protocol. In this protocol, similar to the measurement-only protocols, the verifier has a trusted measurement device. However, in our case, the (trusted) measurements can occur offline, since they are only involved in the verified state preparation stage.

\begin{algorithm}[H]
\caption{One-sided device-independent verification protocol}
 \label{prot:entfk}
 \vskip 0.2 cm
\noindent\textbf{Assumptions} \\
The verifier wants to delegate a quantum computation described by the graph $G$ and specific measurement angles from the set $\{0, \pi/4 ... 7\pi/4 \}$ chosen to define a desired computation. 
For the FK protocol, used in stage $2$, the graph $G$ is encoded as a different graph, $\mathcal{G}$, having $M$ qubits. The verifier has a trusted single-qubit measurement device.\\

\noindent\textbf{Stage 1: Verified state preparation} 
\begin{enumerate}
\item Verifier instructs the server to prepare $\Theta(M^{13} log(M))$ Bell pairs
and send half of each pair to her.
\item Verifier measures a random subset of $M$ qubits in the $\{\Ket{+_{\theta}}, \Ket{-_{\theta}}\}$ basis, where $\theta \in \{0, \pi/4 ... 7\pi/4 \}$, or the computational basis, without instructing the server to measure them.\\
-- For the remaining qubits (which are of order $\Theta(M^{13} log(M))$) she again measures in the $\{\Ket{+_{\theta}}, \Ket{-_{\theta}}\}$ or the computational basis, but she also instructs the server to measure the corresponding entangled partner in the same basis.\\ 
-- The server reports the measurement outcomes and if the results are not the same as the verifiers, the protocol is aborted.
\end{enumerate}

\noindent\textbf{Stage 2: Verified computation} (Robust FK)

\begin{enumerate}
\setcounter{enumi}{2}
\item If the protocol is not aborted in the previous stage, the verifier runs the (robust) FK protocol (given in \cite{gkw}) with the server, using graph $\mathcal{G}$ (i.e. the qubits from the first step of Stage 1 are treated as if they had been sent by the verifier to be entangled into the graph $\mathcal{G}$ and run the FK protocol).
\end{enumerate}
\end{algorithm}

\subsection{Verification based on steering correlations}\label{subsect:new_VDQC}
The setting we have, is that the verifier trusts her measuring device, but the server and the shared state, prepared by the server, are not trusted. 
Since our setting is identical, in terms of trust assumptions, to the one-sided device-independent self-testing scenario, we will sometimes refer to the verifier as Alice and the server as Bob.
If Alice knew that the shared state with Bob is a tensor product of perfect Bell states, she could measure her side and collapse the state of the server to the desired input of FK.
But since the state is untrusted, we will use the results of the previous sections to characterise the shared state from steering correlations. If Alice observes a close to maximal saturation of the steering correlations, she can conclude that the shared state is a tensor product of Bell pairs. This follows from the rigidity result of Theorem~\ref{theorem:rigidity}. The result is analogous to that of \cite{ruv, gkw} where the rigidity of CHSH games is used in the same fashion. All of this is encapsulated in Protocol~\ref{prot:entfk}.

During the verified state preparation stage, the  verifier confirms that all her measurement outcomes agree with the measurement outcomes reported by the server. 
The number of measurement rounds, i.e. $\Theta(M^{13} log(M))$, is chosen so that the results of Theorems~\ref{thm:azuma} and~\ref{theorem:rigidity} lead to a decreasing error. As explained, the theorems require us to have $\Theta(M^{12} log(M))$ rounds of measurement per steering game, and since we have $M$ steering games we end up with $\Theta(M^{13} log(M))$ rounds.
Contrast this to the round complexities in the fully device-independent scenarios which are of order $O(M^c)$, where $c > 8192$ for \cite{ruv} and $c > 2048$ for \cite{gkw}, respectively. Even though these do not follow from tight bounds, it seems clear that the added trust of the verifier's measurement device leads to a significant reduction in complexity.

Next we examine the verified computation stage.
As the input obtained from the previous stage is not exact, but close in trace distance to the ideal state, we need to consider the robust version of the FK derived in \cite{gkw}. Moreover, to obtain the optimal complexity, instead of using a dotted-complete graph as was done in \cite{fk} and \cite{gkw}, we will use the optimised resource construction introduced in \cite{kw}, where the number of qubits in the encoded graph, $\mathcal{G}$, is linear in the number of qubits in the computation graph $G$. This leads us to the following:
\begin{lemma}\label{lemma:1S_FK}
Protocol \ref{prot:entfk}, for a computation of size $M$, utilises $\Theta(M^{13} log(M))$ Bell pairs and has $\Theta(M^{13} log(M))$ round complexity
(rounds of interaction between the server and the verifier).
The probability that the verifier accepts the outcome of the protocol, assuming honest behaviour by the server, is unity.
If $\eta$ is the probability that the verifier accepts an incorrect outcome in the FK protocol (verified computation stage), the overall probability of accepting an incorrect outcome for Protocol~\ref{prot:entfk} is $\eta + \lambda^{-1}$, for some constant $\lambda > 1$ fixed by the verifier.
\end{lemma}
\begin{proof}
As mentioned, the FK protocol with the encoding (resource state) defined in \cite{kw}, or using the resource construction procedure of \cite{kapourniotis2015optimising}, has linear round complexity. Therefore, the round complexity of the verified computation stage is $\Theta(M)$. 
The verified preparation stage requires $\Theta(M^{13} log(M))$ Bell pairs which are used to test the saturation of steering inequalities. Therefore, the overall round complexity is $\Theta(M^{13} log(M))$.
In the honest run, the verified state preparation stage leads to preparing the ideal input, and therefore the probability of accepting an honest run is unity. This is because in the honest setting, the server prepares ideal Bell pairs and all states are measured correctly, leading to the correct rotated qubits on his side, which are then used in the FK protocol. The FK protocol also has probability of acceptance unity, when the server behaves honestly \cite{fk}.

On the other hand, a dishonest run involves deviations in both state preparation and verified computation stages. In the first stage, a dishonest server would prepare states that are at most $O(\epsilon^{1/6})$-deviated from the ideal by Theorem~\ref{thm:azuma}. 
Here, $\epsilon = \Theta(1/M^6)$, in order to have a decreasing error in simulating the ideal strategy.
Thus, the deviation per Bell pair is fixed however, the verifier can choose many measurement rounds for the steering games in the state preparation stage, reducing the deviation in the overall (tensor product) state.
Concretely, if she chooses to run $\lambda^{12} M^6$ steering games, then $\epsilon = \lambda^{-12} M^{-6}$,
for some fixed $\lambda$. The overall deviation, which is of order $M \epsilon^{1/6}$, will then be of order $\lambda^{-2}$.
This constitutes the deviation in the state preparation stage. 
As explained in \cite{gkw}, the robustness of the FK protocol implies that if the computation stage has probability of acceptance, given dishonest behaviour, $\eta$ and there is a deviation of order $\alpha$ in the state preparation phase, then the overall probability of accepting an incorrect outcome will be $\eta + \sqrt{\alpha}$. In our case, $\alpha = \lambda^{-2}$, hence the probability becomes $\eta + \lambda^{-1}$.
\end{proof}

The resulting protocol's round complexity is an improvement, over the device-independent approaches of \cite{gkw, ruv, mckague}. 
The reason for this is that we benefited from the fact that Alice is trusted in the rigidity result. In particular, while there is a cost for removing the i.i.d. assumption, in the semi-trusted setting we have no extra cost to recover the tensor product structure from Theorem \ref{theorem:rigidity}. This leads to a better complexity when we extract the tensor product structure of Bell pairs, than the complexity obtained in the completely untrusted (device-independent) setting of \cite{ruv}.
On the other hand, in the case of FK where the verifier has a trusted preparation device and sends qubits to the server, the round complexity and number of used quantum states are linear in the size of the computation \cite{kw}. 
In our case, the main sources for overhead are:
\begin{itemize}
\item Self-testing with i.i.d. states which gives a tight bound for the closeness of the state up order $\sqrt{\epsilon}$.
\item Gathering sufficient statistics in the non-i.i.d. setting to get a good estimate of the true quantum correlations for the averaged state. This required $\Omega((1/\epsilon^2) log(1/\epsilon))$ rounds of measurements.
\item Inferring the closeness of a typical state from the closeness of the averaged state, which gives a bound of order $\epsilon^{1/6}$.
\end{itemize}

\subsection{Verification from partially entangled states} \label{subsect:totally_steerable}

It is known that if the source of quantum states is trusted, then classical-quantum correlations are sufficient for verification \cite{hpf, fk}. Alternatively, in \emph{online} verification protocols (where computation and quantum communication take place at the same time), it is possible to have an untrusted entanglement source and states which are less than maximally entangled. Our setting is that of \emph{offline} verification with an untrusted entanglement source.

We have seen that in both device-independent verification and one-sided device-independent verification we can characterise the tensor structure of Bell pairs between the verifier and the server from correlations. In both cases, saturating an inequality involving correlations leads to a bound on the trace distance between the shared state and perfect Bell pairs, up to an isometry. While this is sufficient, it does not seem necessary to use Bell states. In fact in QKD and random number generation, as mentioned previously, other types of states can also be used \cite{boundEnt, HHHO, tripartiteQKD, qrngVacuumStates, qrngVacuum} and so it is interesting to examine if this is also the case for verification. Here, interestingly, we show that under reasonable assumptions about the verification protocol, the entangled states must be unitarily equivalent to Bell states.

Quantum steering derives its name from the idea that with an entangled state one party could steer the state of another party through local measurement. The person who performs the steering is untrusted which in our case corresponds to the quantum server. 
The verifier, having a trusted measurement device, can check through local measurement, if the server behaved honestly and steered the state correctly.
We consider such steerable states, $\rho_{AB}$, shared between Alice (the verifier) and Bob (the server), which could be useful for verification in the one-sided device-independent scenario.
Concretely, the specific type of verification we consider is one which is offline, measurement-only and blind. This is akin to the setting of our steering-based verification protocol, however we no longer assume the use of the FK protocol in the verified computation stage. Instead, we assume independent verification, i.e. a black-box type verification protocol obeying the three conditions we previously stated. Moreover, we assume that the verifier wishes to prepare a specific quantum input for her computation. This input is assumed to be as general as possible. Given these conditions, we wish to know what properties are required of a typical bipartite shared state $\rho_{AB}$.
Since we are interested in blind verification, it must be the case that $Tr_A(\rho_{AB}) = \rho_B = I/2$. Additionally, because the verifier needs to prepare her quantum input on the server's side, and this can consist of a varied set of states, we will also assume that $\rho_{AB}$ must be \emph{completely steerable} (i.e. can be steered to any state). We use the definition of complete steering from \cite{pusey}. 
\begin{definition} \cite{pusey}
A bipartite state, $\rho_{AB}$, shared by Alice and Bob is completely steerable by Bob iff for any positive operators $\{ \sigma_a \}$, satisfying $\sum_a \sigma_a = Tr_A(\rho_{AB})$, there exist a POVM $\{ E_a \}$, such that $\sigma_a = Tr_A( (E_a \otimes I) \rho_{AB})$.
\end{definition}
\noindent Adding the blindness assumption to this, we follow up with the definition:
\begin{definition}\label{def:total}
A bipartite state, $\rho_{AB}$, shared by Alice and Bob is totally steerable by Bob iff $\rho_{AB}$ is completely steerable by Bob and $Tr_A(B) = I/2$.
\end{definition}
\noindent Before proving our result, we state a useful lemma from \cite{pusey}, concerning completely steerable states:
\begin{lemma} \label{lemma:complete} \cite{pusey}
A bipartite state, $\rho_{AB}$ is completely steerable by Bob iff there exists a purification $\rho_{ABC}$ such that $\rho_{BC} = \rho_B \otimes \rho_C$, where $\rho_{BC} = Tr_A(\rho_{ABC})$, $\rho_{B} = Tr_{AC}(\rho_{ABC})$, $\rho_{C} = Tr_{AB}(\rho_{ABC})$.
\end{lemma}
\noindent This is similar to QKD, where an adversary can have a purification of the state shared by Alice and Bob however he is uncorrelated with Alice or Bob. This is due to the monogamy of entanglement. Our main theorem concerning $\rho_{AB}$ will justify why this monogamy appears in the case of complete steerability as well. We note that the monogamy of quantum steering has been studied in other works \cite{monogamy1, monogamy2}.
\begin{theorem} \label{theorem:totallysteerable}
A bipartite state, $\rho_{AB}$ is totally steerable by Bob iff $\rho_{AB}$ is maximally entangled.
\end{theorem}
\begin{proof}[Proof sketch]
The proof relies on considering $\rho_{AB}$ in matrix form and expressing constraints on its matrix elements using Lemma~\ref{lemma:complete} and the fact that $\rho_B = I/2$. Solving the system of constraints leads to a density matrix which can easily be shown to correspond to a pure, maximally entangled state. The complete proof is given  in Appendix~\ref{proof:Thm5}.
\end{proof}
This theorem reveals that the $\rho_{AB}$ state is unitarily equivalent to a maximally entangled Bell state. Therefore, the monogamy
of this state is in fact due to the monogamy of Bell states. Importantly, it shows that two-qubit states which are totally steerable are perfect Bell pairs. This implies that if we are to use total steerability in order to perform verification, we must certify Bell states in particular, rather than any other steerable state.
Hence, for the case of offline, measurement-only, blind quantum verification, the quantum resources used for preparing the quantum input, must be unitarily equivalent to Bell states. Intuitively, this can be understood as follows. The verifier wishes to send a certain quantum input to the server. Since this input is as general as possible, the best way to do this is via the general teleportation protocol using Bell states. This also introduces a natural one-time padding to the states, which keeps the protocol blind. It is clear that in order to have perfect fidelity for the teleportation, the shared states must be equivalent to perfect Bell pairs.

\section{Conclusion}\label{sect:conclusions}
We have shown that, in analogy to the rigidity of non-local correlations via CHSH games, we can prove a rigidity property of steering correlations. This allows us to establish a tensor product of Bell pairs in a setting of one trusted party and one untrusted party. However, in the case of steering, the extra cost to obtain a tensor product of Bell states is 
smaller than in the analogous situation for non-local correlations. While in both cases this overhead is polynomial, the additional trust assumptions of the steering setting allowed for a significantly reduced degree for this polynomial.
We arrived at this result by first considering self-testing in the one-sided device-independent setting, using quantum steering correlations. The self-testing result, which makes an i.i.d assumption about the shared state, gives an optimal bound for the trace distance between this shared state and a Bell pair. We then removed the i.i.d. assumption in a generic way which can be applied to any type of self-testing result, allowing for the determination of a single Bell state (one-shot) in a fully adversarial setting. Combining this with a game-based approach for characterising the states and strategies of the two parties, has lead to the rigidity result. It should be noted, that the game-based approach is also used to prove the rigidity of non-local correlations. However, in our setting, because we trust one of the parties it is simpler to obtain a characterisation of the untrusted party's strategy and the shared states. Thus, our result has reduced overhead compared the analogous result for CHSH game rigidity.

We considered an application of rigidity to a verifiable quantum computation protocol. The reason for choosing this particular application is that for verifiable quantum computation the quantum states themselves should be recovered, and thus the necessity to obtain the full tensor product structure of Bell states. This is in contrast to QKD and QRNG where one does not need to recover the quantum state explicitly, rather an information theoretic quantity such as entropy, mutual information, key-rate etc.
Using the rigidity of steering correlations we constructed a one-sided device-independent protocol for verifiable delegated quantum computation. The protocol we obtained has fewer requirements than the fully device-independent protocol thus being closer to a practical application.
Lastly, we have shown that a certain class of states which are useful for verification, totally steerable states, are necessarily maximally entangled. This for example, rules-out the use of local but steerable states \cite{local_steering} for verification. Since establishing the existence of the tensor product structure of maximal entanglement requires the collection of a significant number of statistical samples, this result gives some indication to the general difficulty of quantum verification.

\section*{Acknowledgements}
We would like to thank Damian Markham for useful discussions. We also
acknowledge discussions with Matty J Hoban on his and Ivan \v{S}upi\'c's independent work on self-testing using quantum steering \cite{Hoban16}, that appeared on the arXiv shortly after our preprint. 
We would also like to thank Dominique Unruh and Kristiina Rahkema for reading an earlier draft and suggesting improvements.
AG and PW thank CNRS Telecom ParisTech for their hospitality during a visit
where this work started.
EK acknowledges funding through EPSRC grants EP/N003829/1 and EP/M013243/1.

\appendix

\section{Proof of Theorem~\ref{theorem:selftesting}} \label{sect:selftesting}
In this section we give a complete proof of Theorem~\ref{theorem:selftesting} which characterises one-sided device-independent self-testing. Consider the following isometry:
\begin{equation} \label{eqn:isometryexp}
\Phi(\ket{\psi}) = \frac{1}{2}(I + Y'_B) \Ket{\psi} \Ket{+_y} + \frac{i}{2} X'_B (I - Y'_B) \Ket{\psi} \Ket{-_y}
\end{equation}
An illustration of this isometry is given in Figure~\ref{fig:local_isometry}, where the upper part is Bob's system and the lower part is Alice's system. It should be noted that the control gates act on the target when the control qubit is in the $\Ket{-_y}$ state, and act as identity when the control qubit is $\Ket{+_y}$. This is in contrast to the standard convention in which the control is a computational basis state.
\noindent Here $\Ket{+_y}$ and $\Ket{-_y}$ are the eigenstates of the Pauli $Y$ operator and $P = \frac{1}{\sqrt{2}}(X + Y)$.
We can clearly see that $\Phi = I_A \otimes \Phi_B$, where $\Phi_B$ is determined by the combination of $X'_B$ and $Y'_B$ operators, from expression~\ref{eqn:isometryexp}, which only act on Bob's system.
We proceed to show that when conditions (\ref{eq:th1cond1})-(\ref{eq:th1cond3}) are satisfied, we obtain condition~\ref{eq:selftest}.
First, we show that:
\begin{equation*}\label{ineq:commutation}
||\Phi(M_A N'_B \Ket{\psi}) - M_A N_B \Phi(\Ket{\psi})|| \leq 2\gamma_2
\end{equation*}
Because $M_A$ only acts on Alice's system, whereas the isometry is local on Bob's system, $M_A$ trivially commutes to the left so that $\Phi(M_A N'_B \Ket{\psi}) = M_A \Phi(N'_B \Ket{\psi})$. Now consider the possible choices for $N'_B$. If $N'_B = I$, the relation holds trivially. If $N'_B = Y'_B$, since $Y'_B$ is hermitian and unitary we have that:
\begin{equation}
\Phi(Y'_B\ket{\psi}) = \frac{1}{2}(I + Y'_B) \Ket{\psi} \Ket{+_y} - \frac{i}{2} X'_B (I - Y'_B) \Ket{\psi} \Ket{-_y}
\end{equation}
At the same time, the ideal Pauli operator $Y_B$, acting on Bob's ancilla, has the following effect:
\begin{equation}
Y_B\Phi(\ket{\psi}) = \frac{1}{2}(I + Y'_B) \Ket{\psi} \Ket{+_y} - \frac{i}{2} X'_B (I - Y'_B) \Ket{\psi} \Ket{-_y}
\end{equation}
This is because $Y \Ket{+_y} = \Ket{+_y}$ and $Y \Ket{-_y} = -\Ket{-_y}$ and we notice that the two expressions are identical.
Lastly, when $N'_B = X'_B$ we have:
\begin{eqnarray*} \label{eqn:approxX}
& & \Phi(X'_B\ket{\psi}) = \frac{1}{2}(I + Y'_B) X'_B\Ket{\psi} \Ket{+_y} +\nonumber\\ & & \frac{i}{2} X'_B (I - Y'_B) X'_B \Ket{\psi} \Ket{-_y}
\end{eqnarray*}
And the action of the ideal operator yields:
\begin{equation}
X_B\Phi(\ket{\psi}) = \frac{i}{2}(I + Y'_B) \Ket{\psi} \Ket{-_y} + \frac{1}{2} X'_B (I - Y'_B) \Ket{\psi} \Ket{+_y}
\end{equation}
This is because $X \Ket{+_y} = i\Ket{-_y}$ and $X \Ket{-_y} = (-i)\Ket{+_y}$. Using the approximate anti-commutation of $X'_B$ and $Y'_B$, as given by condition~\ref{eq:th1cond3}, we notice that commuting $X'_B$ to the left in $\Phi(X'_B \Ket{\psi})$ will lead to the same expression as for $X_B \Phi(\Ket{\psi})$ up to $2\gamma_2$ error. Thus:
\begin{equation} \label{ineq:moregeneral}
||\Phi(M_A N'_B \Ket{\psi}) - M_A N_B \Phi(\Ket{\psi})|| \leq 2\gamma_2
\end{equation}
We therefore, only need to examine the closeness of $\Phi(\Ket{\psi})$ to the ideal Bell state tensored with some junk state.
Start by considering the state:
\begin{eqnarray}
\Ket{\phi} &=& \frac{1}{2}(\Ket{\psi} \Ket{+_y} + Y_A \Ket{\psi} \Ket{+_y}+ \nonumber\\ &+& i X_A \Ket{\psi} \Ket{-_y}
- i Y_A X_A \Ket{\psi} \Ket{-_y})
\end{eqnarray}
We will show that $\Phi(\Ket{\psi})$ and $\Ket{\phi}$ are close in trace distance. Firstly, from conditions (\ref{eq:th1cond1})-(\ref{eq:th1cond2}) using suitable triangle inequalities and the unitarity of operators $X_A$ and $Y_A$ which do not increase trace distance, it can be shown that:
\begin{equation*}
|| (X'_B Y'_B - Y_A X_A) \Ket{\psi} || \leq 2\gamma_1
\end{equation*}
Expanding the trace distance of $\Phi(\Ket{\psi})$ and $\Ket{\phi}$ we have:
\begin{eqnarray*}
& &||\Phi(\ket{\psi}) - \Ket{\phi}|| = \frac{1}{2}|| (Y'_B - Y_A) \Ket{\psi} \Ket{+_y}+\\ & & + i (X'_B - X_A) \Ket{\psi} \Ket{-_y}
- i (X'_B Y'_B - Y_A X_A) \Ket{\psi} \Ket{-_y} ||\nonumber
\end{eqnarray*}
And using the above results it follows that:
\begin{equation*}
||\Phi(\ket{\psi}) - \Ket{\phi}|| \leq 2 \gamma_1
\end{equation*}
Let us now rewrite $\Ket{\phi}$. Given that we trust Alice's side of the $\Ket{\psi}$ state, we can express it as follows:
\begin{equation*}
\Ket{\psi} = a \Ket{\alpha}_B \Ket{+_y}_A + b \Ket{\beta}_B \Ket{-_y}_A
\end{equation*}
Where $|a|^2 + |b|^2 = 1$ and the states $\Ket{\alpha}$ and $\Ket{\beta}$ are normalized. 
Here, the first part denotes Bob's system, for which we can make no assumptions, and the second part is Alice's qubit. 
The reason for choosing Pauli-$Y$ eigenstates on Alice's side is to simplify the calculation. We could have expanded her system in any basis since a local unitary on her system does not change the result.
Substituting this into the $\Ket{\phi}$ and labeling the ancillary qubit introduced by this isometry with the label $\Phi$ we get:
\begin{eqnarray}
	 \ket{\phi} & = & 
	 \frac{1}{2} ( a \Ket{+_y}_{\Phi} \Ket{\alpha}_B \Ket{+_y}_A +  b \Ket{+_y}_{\Phi} \Ket{\alpha}_B \Ket{-_y}_A ) \nonumber\\
	 & + & \frac{1}{2} Y_A( a \Ket{+_y}_{\Phi} \Ket{\alpha}_B \Ket{+_y}_A +  b \Ket{+_y}_{\Phi} \Ket{\alpha}_B \Ket{-_y}_A ) \nonumber\\
	 & + & \frac{1}{2} X_A( a \Ket{-_y}_{\Phi} \Ket{\alpha}_B \Ket{+_y}_A +  b \Ket{-_y}_{\Phi} \Ket{\alpha}_B \Ket{-_y}_A ) \nonumber\\
	 & - & \frac{1}{2} Y_A X_A( a \Ket{-_y}_{\Phi} \Ket{\alpha}_B \Ket{+_y}_A +  b \Ket{-_y}_{\Phi} \Ket{\alpha}_B \Ket{-_y}_A )\nonumber\\ & & 	
\end{eqnarray}
Using the following identities:
\begin{equation*}
X \Ket{+_y} = i \Ket{-_y} \; \; \; X \Ket{-_y} = -i \Ket{+_y}
\end{equation*}
\begin{equation*}
Y \Ket{+_y} = \Ket{+_y} \; \; \; Y \Ket{-_y} = - \Ket{-_y}
\end{equation*}
We reduce $\Ket{\phi}$ to:
\begin{eqnarray}
	 \ket{\phi} & = & 
	 \frac{1}{2} ( a \Ket{+_y}_{\Phi} \Ket{\alpha}_B \Ket{+_y}_A +  b \Ket{+_y}_{\Phi} \Ket{\alpha}_B \Ket{-_y}_A ) \nonumber\\
	 & + & \frac{1}{2} ( a \Ket{+_y}_{\Phi} \Ket{\alpha}_B \Ket{+_y}_A -  b \Ket{+_y}_{\Phi} \Ket{\alpha}_B \Ket{-_y}_A ) \nonumber\\
	 & - & \frac{1}{2} ( a \Ket{-_y}_{\Phi} \Ket{\alpha}_B \Ket{-_y}_A -  b \Ket{-_y}_{\Phi} \Ket{\alpha}_B \Ket{+_y}_A ) \nonumber\\
	 & - & \frac{1}{2} ( a \Ket{-_y}_{\Phi} \Ket{\alpha}_B \Ket{-_y}_A +  b \Ket{-_y}_{\Phi} \Ket{\alpha}_B \Ket{+_y}_A )\nonumber\\ & &
\end{eqnarray}
The terms with $b$ coefficient cancel out and we are left with:
\begin{equation*}
\Ket{\phi} = a \Ket{\alpha}_B (\Ket{+_y}_{\Phi}\Ket{+_y}_A - \Ket{-_y}_{\Phi} \Ket{-_y}_A)
\end{equation*}
This state is equivalent to:
\begin{equation*}
\Ket{\phi} = a\sqrt{2} \Ket{\alpha}_B \Ket{\psi_+}_{AB}
\end{equation*}
We would like to equate this to $\Ket{junk}_B \Ket{\psi_+}_{AB}$, however, the state we have is unnormalized unless $a = 1/\sqrt{2}$.
We therefore compute a bound on $|a|$ to determine the error introduced by the unnormalized state.
Condition~\ref{eq:th1cond1} can be rewritten as:
\begin{equation*}
1 - \gamma_1^2/2 \leq  \Bra{\psi} X_A X'_B \Ket{\psi}
\end{equation*}
By expanding $\Ket{\psi}$, applying the operators $X_A$, $X'_B$ and using the facts that $|a|^2 + |b|^2 = 1$ and that $X'_B$ is hermitian and has $\pm 1$ eigenvalues, we obtain:
\begin{equation*}
\sqrt{1 - \gamma_1^2/2} \leq a \sqrt{2} \leq \sqrt{1 + \gamma_1^2/2}
\end{equation*}
And since for small $\gamma_1$ we know that $\sqrt{1 - \gamma_1^2/2}$ approaches $1 - \gamma_1^2/4$ and $\sqrt{1 + \gamma_1^2/2}$ approaches $1 + \gamma_1^2/4$ we have that the norm of $\Ket{\phi}$, can change from unity by an order of $\gamma_1^2/4$. Thus, it follows that:
\begin{equation*}
|| \Phi(\Ket{\psi}) - \Ket{junk}_B \Ket{\psi_+}_{AB} || \leq 3 \gamma_1 + \gamma_1^2/4
\end{equation*}
Lastly, together with inequality~\ref{ineq:moregeneral} and a triangle inequality, we get:
\begin{equation}
|| \Phi(M_A N'_B\Ket{\psi}) - \Ket{junk}_B M_A N_B\Ket{\psi_+}_{AB} ||  \leq 3 \gamma_1 + \gamma_1^2/4 + 2\gamma_2
\end{equation}

\section{Proof of Theorem~\ref{theorem:steeringrobust}} \label{sect:steeringrobust}
Theorem~\ref{theorem:steeringrobust} shows that saturating the correlation of observables on Alice and Bob's side, with Alice being trusted, leads to the necessary conditions of Theorem~\ref{theorem:selftesting} which imply that the shared state is close, up to local isometry, to a Bell state. Similar to Theorem~\ref{theorem:selftesting} we start the proof by denoting $B_0$ as $X'_B$ and $B_1$ as $Y'_B$.
Splitting equation~\ref{eqn:saturate}, we have:
\begin{equation} \label{eq:split}
\bra{\psi}   X_A X'_B \ket{\psi} + \bra{\psi} Y_A Y'_B \ket{\psi} \geq 2 -\epsilon
\end{equation}
However, it's clear that:
\begin{equation*}
-1 \leq \bra{\psi}   X_A X'_B \ket{\psi} \leq 1
\end{equation*}
\begin{equation*}
-1 \leq \bra{\psi}   Y_A Y'_B \ket{\psi} \leq 1
\end{equation*}
So, it follows that:
\begin{equation*}
\bra{\psi}   X_A X'_B \ket{\psi} \geq 1 - \epsilon
\end{equation*}
\begin{equation*}
\bra{\psi}   Y_A Y'_B \ket{\psi} \geq 1 - \epsilon
\end{equation*}
This allows us to derive conditions~\ref{eq:th1cond1} and~\ref{eq:th1cond2}, since:
\begin{equation*}
||(X_{A} -  X^{\prime}_{B}) \ket{\psi}|| = \sqrt{2} \sqrt{1 - \bra{\psi} X_A X'_B \ket{\psi} } \leq \sqrt{2 \epsilon}
\end{equation*}
\begin{equation*}
||(Y_{A} -  Y^{\prime}_{B}) \ket{\psi}|| = \sqrt{2} \sqrt{1 - \bra{\psi} Y_A Y'_B \ket{\psi} } \leq \sqrt{2 \epsilon}
\end{equation*}
Hence, in Theorem~\ref{theorem:selftesting}, $\gamma_1 = \sqrt{2\epsilon}$.
Let us now denote:
\begin{equation*}
S =  X_A X'_B + Y_A Y'_B
\end{equation*}
Computing $S^2$ and using the fact that $X_A Y_A = iZ_A$ we obtain:
\begin{equation*}
S^2 = 2 + iZ_A [ X'_B, Y'_B ]
\end{equation*}
Since $[X_A, Y_A] = 2iZ_A$, we can alternatively write this as:
\begin{equation*}
S^2 = 2 + \frac{1}{2}[X_A, Y_A] [ X'_B, Y'_B ]
\end{equation*}
The Cauchy-Schwarz inequality together with inequality~\ref{eq:split} give us:
\begin{equation*}
 \bra{\psi} S^2 \ket{\psi}  \geq | \bra{\psi} S \ket{\psi} |^2 \geq (2 - \epsilon)^2
\end{equation*}
Substituting $S^2$:
\begin{equation*}
\Bra{\psi} 2 + \frac{1}{2}[X_A, Y_A] [ X'_B, Y'_B ] \Ket{\psi} \geq 4 - 4\epsilon + \epsilon^2 \geq 4 - 4\epsilon
\end{equation*}
Hence:
\begin{equation*}
\Bra{\psi} [X_A, Y_A] [ X'_B, Y'_B ] \Ket{\psi} \geq 4 - 8\epsilon
\end{equation*}
Expanding the commutators yields:
\begin{eqnarray*}
& &\Bra{\psi} X_A Y_A X'_B Y'_B \Ket{\psi} - \Bra{\psi} X_A Y_A Y'_B X'_B \Ket{\psi}\nonumber\\ & &- \Bra{\psi} Y_A X_A X'_B Y'_B \Ket{\psi}+ \Bra{\psi} Y_A X_A Y'_B X'_B \Ket{\psi} \geq 4 - 8\epsilon
\end{eqnarray*}
By splitting into terms, as we did with inequality~\ref{eq:split}, we have that:
\begin{equation}\label{eqn:comm1}
\bra{\psi} X_A Y_A X'_B Y'_B \ket{\psi} \geq 1 - 8\epsilon
\end{equation}
\begin{equation}\label{eqn:comm2}
\bra{\psi} Y_A X_A Y'_B X'_B \ket{\psi} \geq 1 - 8\epsilon
\end{equation}
\begin{equation}\label{eqn:comm3}
\bra{\psi} X_A Y_A Y'_B X'_B \ket{\psi} \leq 8\epsilon - 1
\end{equation}
\begin{equation}\label{eqn:comm4}
\bra{\psi} Y_A X_A X'_B Y'_B \ket{\psi} \leq 8\epsilon - 1
\end{equation}
Now using the fact that $X_A Y_A + Y_A X_A = 0$, we have:
\begin{eqnarray*}
& &||( X^{\prime}_{B} Y^{\prime}_{B} +  Y^{\prime}_{B} X^{\prime}_{B}) \ket{\psi}|| =\nonumber\\ & &||( X^{\prime}_{B} Y^{\prime}_{B} + X_A Y_A + Y_A X_A + Y^{\prime}_{B} X^{\prime}_{B}) \ket{\psi}||
\end{eqnarray*}
And using a triangle inequality, we have:
\begin{eqnarray*}
& &||( X^{\prime}_{B} Y^{\prime}_{B} +  Y^{\prime}_{B} X^{\prime}_{B}) \ket{\psi}|| \leq \nonumber\\ & &||( X_A Y_A +  X^{\prime}_{B} Y^{\prime}_{B}) \ket{\psi}|| + ||( Y_A X_A +  Y^{\prime}_{B} X^{\prime}_{B}) \ket{\psi}||
\end{eqnarray*}
Additionally:
\begin{eqnarray*}
& &||( X_A Y_A +  X^{\prime}_{B} Y^{\prime}_{B}) \ket{\psi}|| =\nonumber\\ & & \sqrt{2 + \bra{\psi} X_A Y_A Y'_B X'_B \ket{\psi} + \bra{\psi} Y_A X_A X'_B Y'_B \ket{\psi}}
\end{eqnarray*}
And from inequalities~\ref{eqn:comm3},~\ref{eqn:comm4} we get that:
\begin{equation*}
||( X_A Y_A +  X^{\prime}_{B} Y^{\prime}_{B}) \ket{\psi}|| \leq 4 \sqrt{\epsilon}
\end{equation*}
Similarly, using inequalities~\ref{eqn:comm1},~\ref{eqn:comm2}, we have:
\begin{equation*}
||( Y_A X_A +  Y^{\prime}_{B} X^{\prime}_{B}) \ket{\psi}|| \leq 4 \sqrt{\epsilon}
\end{equation*}
Which leads to:
\begin{equation*}
||( X^{\prime}_{B} Y^{\prime}_{B} +  Y^{\prime}_{B} X^{\prime}_{B}) \ket{\psi}|| \leq 8 \sqrt{\epsilon}
\end{equation*}
Thus satisfying condition~\ref{eq:th1cond3}, with $\gamma_2 = 8 \sqrt{\epsilon}$, and concluding the proof.

\section{Proof of Theorem~\ref{theorem:tightbound}} \label{sect:tightbound}
Theorems~\ref{theorem:selftesting} and~\ref{theorem:steeringrobust} show that if the correlation of local observables is saturated up to order $O(\epsilon)$, the shared state is close, up to local isometry, to a Bell state up to order $O(\sqrt{\epsilon})$. Theorem~\ref{theorem:tightbound} shows that this bound is tight, up to constant factors.
We prove this theorem by contradiction. Assume the bound of theorem~\ref{theorem:selftesting} is not tight and it is possible to derive an asymptotically better bound for the shared state of Alice and Bob. In particular, this means that there is no state $\Ket{\psi}$ which is $O(\sqrt{\epsilon})$-close to the $\Ket{\psi_+}$ state and there are no observables $B_0$ and $B_1$ such that inequality~\ref{eqn:saturate2} is satisfied. However, letting $\epsilon' = \epsilon/2$, consider the following state:
\begin{equation*}
\Ket{\psi} = \frac{1}{\sqrt{2}}\left(\sqrt{1 + \sqrt{\epsilon'}} \Ket{01} + \sqrt{1 - \sqrt{\epsilon'}}\Ket{10}\right)
\end{equation*}
We have that:
\begin{equation*}
||\Ket{\psi} - \Ket{\psi_+} || = \sqrt{2 - \Braket{\psi|\psi_+} - \Braket{\psi_+|\psi}}
\end{equation*}
Notice that:
\begin{equation*}
\Braket{\psi|\psi_+} = \Braket{\psi_+|\psi} = \frac{1}{2}\left(\sqrt{1 + \sqrt{\epsilon'}} + \sqrt{1 - \sqrt{\epsilon'}}\right)
\end{equation*}
Substituting this into the previous expression, and taking the first order term we have:
\begin{equation*}
||\Ket{\psi} - \Ket{\psi_+} || = O(\sqrt{\epsilon'})
\end{equation*}
Consider also the observables:
\begin{eqnarray*}
B_0 &=& \left( \begin{array}{ccc}
-\sqrt{\epsilon'} & \sqrt{1  - \epsilon'} \\
\sqrt{1 - \epsilon'} & \sqrt{\epsilon'} \\
\end{array} \right) \; \;
\nonumber\\ B_1 &=& \left( \begin{array}{ccc}
0 & \sqrt{\epsilon'} - i\sqrt{1 - \epsilon'} \\
\sqrt{\epsilon'} + i\sqrt{1 - \epsilon'} & 0 \\
\end{array} \right)
\end{eqnarray*}
One can check that $B_0 = B^{\dagger}_0$, $B_1 = B^{\dagger}_1$, $B_0 B^{\dagger}_0 = B_1 B^{\dagger}_1 = I$ and that the two matrices have eigenvalues $\pm 1$. Moreover, we can see that as $\epsilon' \rightarrow 0$ we have that $B_0 \rightarrow X$ and $B_1 \rightarrow Y$. Importantly, we have that:
\begin{equation*}
\bra{\psi} X_A B_0 \Ket{\psi} = \Bra{\psi} Y_A B_1 \Ket{\psi} = 1 - \epsilon'
\end{equation*}
And therefore:
\begin{equation*}
\bra{\psi}\left(   X_A B_0 + Y_A B_1     \right)  \ket{\psi} = 2 -2\epsilon' = 2 - \epsilon
\end{equation*}
Thus, inequality~\ref{eqn:saturate2} is saturated. But this should not be possible under the assumption that the bound on $\Ket{\psi}$'s closeness to $\Ket{\psi_+}$ is not tight. Therefore, the assumption is false and the $O(\sqrt{\epsilon})$ bound is tight for this type of steering inequality.
Note that this result still holds under local isometry since the isometry is, by definition, distance preserving and so under the local isometry the state is still $O(\sqrt{\epsilon})$-close to a Bell pair.

\section{Proof of Theorem~\ref{thm:azuma}} \label{sect:azuma}
Theorem~\ref{thm:azuma} shows that from the observed outcomes of Alice and Bob, after $K$ rounds of measurement, we can conclude something about their shared quantum state in a single round, even without assuming independence. Simply stated, if the shared state of Alice and Bob is $\sigma$, we need not assume that $\sigma$ is a tensor product of identical states. Nevertheless, denoting the state of round $i$ as $\rho_i = Tr_{-i}(\mathcal{E}^{AB}_{1,i-1}(\sigma))$, from their shared correlations we can deduce that $\rho_i$ is close in trace distance to a perfect Bell state (under a suitable isometry).
The proof consists of a number of steps:
\begin{enumerate}
\item Firstly, we show that the observed correlations of Alice and Bob's, given fixed measurement settings, is a good estimate for the true quantum correlation assuming they shared multiple copies of the averaged state $\rho_{avg} = \frac{1}{K}\sum\limits_{i=1}^{K} \rho_i$.
\item Secondly, we use the previous result to estimate the correlations for the two measurement settings under consideration. We then use self-testing to show that if the correlations are close to the maximal value, the averaged state is close to a Bell state, under a suitable local isometry.
\item Lastly, we prove that if the averaged state is close to a Bell state (possibly in tensor product with some mixed state), then a typical state $\rho_i$ is also close to that pure state and we compute the exact bound for the trace distance.
\end{enumerate}
We start by proving the first step:
\begin{lemma} \label{lemma:martingales}
Assume Alice and Bob are asked to perform $n$ rounds of measurement of the two-outcome observables with $\pm 1$ eigenvalues, $A$ and $B$, respectively. We denote the outcomes of their measurements as $\{a_i\}$ and $\{ b_i \}$ and $\hat{C}_i = a_i b_i$ as their correlation for round $i$. Additionally, let $H_i = \{ (a_j, b_j) | j < i \}$ be the history of their measurement outcomes up to, but not including round $i$. Finally, letting $C_i = E(\hat{C}_i | H_i)$ to be the conditional expectation value of the correlation given the previous history of outcomes, we have that for any $\delta > 0$:
\begin{equation*}
Pr \left( \left\lvert \frac{1}{n} \sum\limits_{i=1}^{n} C_i - \frac{1}{n} \sum\limits_{i=1}^{n} \hat{C}_i \right\rvert \geq \delta \right) \leq exp(-\delta^2 n / 8)
\end{equation*}
\end{lemma}
\begin{proof}
The variable $C_i$ represents the true correlation of the outcomes in round $i$, as determined by the shared state of Alice and Bob. If the shared state in round $i$ is $\rho_i$ then $C_i = Tr(AB \; \rho_i)$. As mentioned, while we trust Alice and know that she is indeed measuring the observable $A$, we can still assume that Bob is measuring the observable $B$ in each round. This is because the observable $B$ is unrestricted (apart from being a two-outcome observable) and can in principle act on Bob's ancilla as well. Furthermore, we make no assumption about the state $\rho_i$, since it is prepared by Bob.
Another way in which we can express $C_i$ is using its definition, which leads us to:
\begin{equation*}
C_i = Pr(a_i = b_i | H_i) - Pr(a_i \neq b_i | H_i)
\end{equation*}
We now define the random variables:
\begin{equation*}
X_j = \sum\limits_{i = 1}^{j} (C_i - \hat{C}_i)
\end{equation*}
Notice that for any $j \leq n$, $ | X_{j+1} - X_j | \leq 2$ (because $\hat{C}_i = \pm 1$, $-1 \leq C_i \leq 1$), $E(X_j) \leq \infty$ and:
\begin{equation*}
E(X_{j+1} - X_j | H_{j+1}) = C_{j+1} - C_{j+1} = 0
\end{equation*}
Therefore, $\{ X_j \}$ forms a martingale. We can therefore apply the Azuma-Hoeffding inequality \cite{azuma}, in a manner analogous to \cite{Pironio2010, hpf}. Setting $j = n$, we have that for any $t > 0$:
\begin{equation*}
Pr ( | X_n | > t) \leq exp(-t^2 / 8n)
\end{equation*}
Expanding, we have that:
\begin{equation*}
Pr \left( \left\lvert \sum\limits_{i = 1}^{n} (C_i - \hat{C}_i) \right\rvert > t \right) \leq exp(-t^2 / 8n)
\end{equation*}
For some $\delta > 0$, let $t = n \delta$. This yields:
\begin{equation*}
Pr \left( \left\lvert \frac{1}{n} \sum\limits_{i = 1}^{n} C_i - \frac{1}{n} \sum\limits_{i = 1}^{n} \hat{C}_i \right\rvert > \delta \right) \leq exp(-\delta^2 n / 8)
\end{equation*}
Thus concluding the proof of Lemma~\ref{lemma:martingales}.
\end{proof}
We move on to the second step:
\begin{lemma}\label{lemma:average}
Suppose Alice and Bob are required to perform $K$ rounds of measurement and also that:
\begin{itemize}
\item Prior to the measurements, the shared state of Alice and Bob is assumed to be $\sigma$.
\item Alice chooses a random set of size $K/2$, consisting of indices from $1$ to $K$ and denoted $R_0 = \{ i | i \in_{R} \{1 ... K\} \}$, $|R_0| = K/2$. We also denote $R_1 = \{1 ... K \} \backslash R_0$, to be the complement of $R_0$.
\item We denote $\rho_i = Tr_{-i}(\mathcal{E}^{AB}_{1,i-1}(\sigma))$ the reduced state of Alice and Bob in round $i$, and $\rho_{avg} =
\frac{1}{K} \sum\limits_{i = 1}^{K} \rho_i$ as the averaged state. Here $\mathcal{E}^{AB}_{1,i-1}$ denotes the action (measurements) of Alice and Bob on the state $\sigma$ up to round $i$.
\item In round $i$, let $r_i = 0$ iff $i \in R_0$, otherwise $r_i = 1$. Alice measures the observable $A_{r_i}$ on her half of $\rho_i$. $A_0$ and $A_1$ are anti-commuting single-qubit observables having $\pm 1$ eigenvalues.
\item In round $i$, Bob is asked to measure $B_{r_i}$. $B_0$ and $B_1$ have $\pm 1$ eigenvalues.
\item We denote $a_i$ and $b_i$, respectively, as the outcomes of their measurements in round $i$. We also denote $\hat{C}_i = a_i b_i$ as their correlation for round $i$.
\item We denote $\hat{C}^0 = \frac{1}{K/2}\sum\limits_{i \in R_0} \hat{C}_i$ and $\hat{C}^1 = \frac{1}{K/2}\sum\limits_{i \in R_1} \hat{C}_i$ as the averaged correlations for the cases where both Alice and Bob are asked to measure the first observable, or both are asked to measure the second, respectively.
\end{itemize}
If, for some given $\epsilon > 0$ and suitably chosen $K = \Omega((1/\epsilon^2) log(1/\epsilon))$, it is the case that $ \hat{C}^0  +  \hat{C}^1  \geq 2 - \epsilon$ (or, alternatively, $ \hat{C}^0  +  \hat{C}^1  \leq -2 + \epsilon$)  then there exists an isometry $\Phi$ and a mixed state $\rho_{junk}$ such that:
\begin{equation*}
|| \Phi( \mathcal{E}^{AB}(\rho_{avg}) ) -  \hat{\mathcal{E}}^{AB}(\Ket{\psi_+}\Bra{\psi_+}) \rho_{junk} || \leq O(\sqrt{\epsilon})
\end{equation*}
Where $\mathcal{E}^{AB}$ is some combination of the $A_0, A_1, B_0, B_1$ operators and $\hat{\mathcal{E}}^{AB}$ is the analogous combination of the ideal operators $I, X, Y$, as in Theorem~\ref{theorem:steeringrobust}.
\end{lemma}
\begin{proof}
We will prove the case $\hat{C}^0  +  \hat{C}^1  \geq 2 - \epsilon$ since for $\hat{C}^0  +  \hat{C}^1  \leq -2 + \epsilon$ the derivation is similar. Additionally, we only consider the case $\mathcal{E}^{AB} = I$, since the other cases follow from the linearity of the operators.
The previous Lemma, essentially shows us that the observed average correlation is a good estimate for the average true correlation.
Specifically, it is the case that $\hat{C}^b$, $b \in \{0, 1\}$, is close to the quantum correlation $Tr(A_b B_b  \rho_{avg})$.
Consider now a state $\Ket{\zeta}$ which is a purification of $\rho_{avg}$. We can then write the quantum correlation as $\Bra{\zeta} A_b B_b \Ket{\zeta}$. Using these results, if our estimate of the true correlation is of precision (closeness) $\delta > 0$, then it is the case that:
\begin{equation*}
\Bra{\zeta} A_0 B_0 \Ket{\zeta} + \Bra{\zeta} A_1 B_1 \Ket{\zeta} \geq 2 - \epsilon - \delta
\end{equation*}
With probability $1 - exp(-\delta^2 K / 16)$. Let $\delta = \epsilon$ so that we have:
\begin{equation*}
\Bra{\zeta} A_0 B_0 \Ket{\zeta} + \Bra{\zeta} A_1 B_1 \Ket{\zeta} \geq 2 - O(\epsilon)
\end{equation*}
Using Theorem~\ref{theorem:steeringrobust}, it follows that there exists a local isometry $\Phi$ and a state $\Ket{junk}$ such that, with probability $1 - exp(-\epsilon^2 K / 16)$, we have:
\begin{equation*}
|| \Phi( \Ket{\zeta} ) -  \Ket{\psi_+} \Ket{junk} || \leq O(\sqrt{\epsilon})
\end{equation*}
This also implies:
\begin{equation*}
TD(\Phi( \Ket{\zeta} ),  \Ket{\psi_+} \Ket{junk}) \leq O(\sqrt{\epsilon})
\end{equation*}
As mentioned, we are only considering the case of $I$ acting on the state $\Ket{\zeta}$. Of course, the argument proceeds identically, when considering $M_A N'_B \Ket{\zeta}$, as in Theorem~\ref{theorem:selftesting}, leading to the $\mathcal{E}^{AB} \neq I$ cases.
It should be noted that from the construction of $\Phi$ (in Theorem~\ref{theorem:selftesting}), in the case where the shared state is a purification of some mixed state (as is the case with $\Ket{\zeta}$ and $\rho_{avg}$), the isometry does not act on the quantum states used for purification. Therefore, we can trace out those states, and since this operation cannot increase trace distance we have that:
\begin{equation*}
TD( \Phi( \rho_{avg} ),  \Ket{\psi_+}\Bra{\psi_+} \rho_{junk}) \leq O(\sqrt{\epsilon})
\end{equation*}
With probability $1 - exp(-\epsilon^2 K / 16)$. We can incorporate this probability into the trace distance, and we have that:
\begin{equation*}
TD( \Phi( \rho_{avg} ),  \Ket{\psi_+}\Bra{\psi_+} \rho_{junk} ) \leq
O(\sqrt{\epsilon}) + exp(-\epsilon^2 K / 16)
\end{equation*}
Setting $K = -(16/\epsilon^2) log(\sqrt{\epsilon}) = (8/\epsilon^2)log(1/\epsilon)$ we are left with:
\begin{equation*}
TD( \Phi( \rho_{avg} ),  \Ket{\psi_+}\Bra{\psi_+} \rho_{junk}) \leq O(\sqrt{\epsilon})
\end{equation*}
Analogously, we get:
\begin{equation*}
TD( \Phi( \mathcal{E}^{AB}(\rho_{avg}) ),  \hat{\mathcal{E}}^{AB}(\Ket{\psi_+}\Bra{\psi_+}) \rho_{junk}) \leq O(\sqrt{\epsilon})
\end{equation*}
Concluding the proof of Lemma~\ref{lemma:average}.
\end{proof}
The final step towards proving Theorem~\ref{thm:azuma} is to prove the following statement:
\begin{lemma}\label{lemma:closeness}
Let $\rho$ be a general quantum state which we can write as $\rho = \frac{1}{n} \sum\limits_{i=1}^{n} \rho_i$, for some $n > 0$ and density matrices $\rho_i$. If it is the case that, for some pure state $\Ket{\phi}$ and $\gamma > 0$:
\begin{equation*}
TD( \rho, \Ket{\phi}\Bra{\phi}) \leq \gamma
\end{equation*}
Then for a uniformly at random chosen $i \in_R \{1 ... n \}$, with probability at least $1 - \gamma^{1/3}$ we have that:
\begin{equation*}
TD( \rho_i, \Ket{\phi}\Bra{\phi}) \leq O(\gamma^{1/3})
\end{equation*}
\end{lemma}
\begin{proof}
Starting with the bound on the trace distance and the relation between trace distance and fidelity (when one state is pure), we obtain:
\begin{eqnarray}\label{eq:average1}
1-F^2(\rho,\ket{\phi})&\leq&TD(\rho,\ket{\phi}\bra{\phi})\leq\gamma
\nonumber\\
1-\gamma&\leq& |\bra{\phi}\rho\ket{\phi}|=F^2(\rho,\ket{\phi})\nonumber\\
&\leq&\frac1n\sum_i F^2(\rho_i,\ket{\phi}):=\frac1n\sum_i q_i
\end{eqnarray}
where the second inequality follows from convexity and for convenience we defined $q_i=F^2(\rho_i,\ket{\phi})$. This gives a lower bound on the average $q_i$ (average fidelity squared). To provide an upper bound on the average $q_i$, we do the following. We let $p$ be the fraction of $i$'s such that $q_i\leq 1-\gamma-\delta$, where $\delta\in[0,1-\gamma]$. Since $q_i\leq 1$ it follows:
\begin{eqnarray} \label{eq:average2}
(1-\gamma-\delta)\cdot p+1\cdot (1-p)&\geq&\frac1n\sum_i q_i
\end{eqnarray}
From Equations (\ref{eq:average1},\ref{eq:average2}) it follows that
\begin{eqnarray*}
(1-\gamma-\delta)\cdot p+1\cdot (1-p)&\geq& 1-\gamma\nonumber\\
\frac{\gamma}{\gamma+\delta}&\geq& p
\end{eqnarray*}
Now, using the fact that $TD(\rho_i,\Ket{\phi}\Bra{\phi})\leq \sqrt{1-q_i}$, we note that with probability $(1-p)$ we have that $q_i\geq 1-\gamma-\delta$ and thus for these cases $TD(\rho_i,\Ket{\phi}\Bra{\phi})\leq \sqrt{\gamma+\delta}$. By choosing $\delta=\gamma^{2/3}-\gamma$ we have that with probability at least $1 - \gamma^{1/3}$:
\begin{equation*}
TD(\rho_i,\Ket{\phi}\Bra{\phi})\leq\sqrt{\delta+\gamma}=\gamma^{1/3}
\end{equation*}
\end{proof}
We can now use these lemmas to prove Theorem~\ref{thm:azuma}.
\begin{proof}[Proof of Theorem~\ref{thm:azuma}]
We have the same assumptions as in Lemma~\ref{lemma:average} and from it we know that after $K = \Omega((1/\epsilon^2) log(1/\epsilon))$ rounds of measurements, if the observed correlations are saturated up to order $\epsilon$ we have:
\begin{equation*}
TD( \Phi(\mathcal{E}^{AB}(\rho_{avg})), \hat{\mathcal{E}}^{AB}(\Ket{\psi_+} \Bra{\psi_+}) \rho_{junk}) \leq O(\sqrt{\epsilon})
\end{equation*}
We would now like to apply Lemma~\ref{lemma:closeness}, however because of $\rho_{junk}$ we do not have a pure state in the trace distance expression. Therefore, we trace out the junk subsystem, and since tracing out can only decrease trace distance we have that:
\begin{equation*}
TD( Tr_{junk}(\Phi(\mathcal{E}^{AB}(\rho_{avg}))), \hat{\mathcal{E}}^{AB}(\Ket{\psi_+} \Bra{\psi_+})) \leq O(\sqrt{\epsilon})
\end{equation*}
If we denote $\rho = Tr_{junk}(\Phi(\mathcal{E}^{AB}(\rho_{avg})))$, we can now apply Lemma~\ref{lemma:closeness}. Note that $Tr_{junk}(\Phi(\mathcal{E}^{AB}(\cdot)))$ is a linear map, since it is the composition of linear maps. Therefore, it is the case that:
\begin{equation*}
Tr_{junk}(\Phi(\mathcal{E}^{AB}(\rho_{avg}))) = \sum\limits_{i=1}^{K} Tr_{junk}(\Phi(\mathcal{E}^{AB}(\rho_{i})))
\end{equation*}
Hence, for a randomly chosen $\rho_i$, with probability at least $1 - O(\epsilon^{1/6})$:
\begin{equation*}
TD( Tr_{junk}(\Phi(\mathcal{E}^{AB}(\rho_{i}))), \hat{\mathcal{E}}^{AB}(\Ket{\psi_+} \Bra{\psi_+})) \leq O(\epsilon^{1/6})
\end{equation*}
Thus we have shown that, under a suitably chosen local isometry and tracing out any additional systems we can ``extract'' a single Bell pair from the two device's shared state.
If we want to also have a closeness relation for $\Phi(\mathcal{E}^{AB}(\rho_i))$, which includes all of Bob's private subsystem, we can use a corollary of the Gentle Measurement Lemma:

\begin{corollary}
\cite{ruv, gkw} Let $\rho$ be a state on $\mathcal{H}_1 \otimes \mathcal{H}_2$, and let $\pi$ be a pure state on $\mathcal{H}_1$.
If for some $\delta \geq 0$, $Tr (\pi Tr_2( \rho)) \geq 1 - \delta$, then
\begin{equation*}
TD( \rho, \pi \otimes Tr_1( \rho)) \leq 2 \sqrt \delta + \delta
\end{equation*}
\end{corollary}
\noindent This leads to:
\begin{equation*}
TD( \Phi(\mathcal{E}^{AB}(\rho_{i})), \hat{\mathcal{E}}^{AB}(\Ket{\psi_+} \Bra{\psi_+}) \tilde{\rho}_{junk}) \leq O(\epsilon^{1/12})
\end{equation*}
Where $\tilde{\rho}_{junk} = Tr_{-junk}(\Phi(\mathcal{E}^{AB}(\rho_{i})))$.
\end{proof}

\section{Proof of Theorem~\ref{theorem:rigidity}} \label{sect:rigidity}
As mentioned, we prove Theorem~\ref{theorem:rigidity} by showing that the actual strategy $\mathcal{S}$ of Alice and Bob is close to a strategy $\mathcal{S}_{g}$ in which Alice guesses Bob's outcomes, which in turn is close to the ideal strategy $\mathcal{S}_{id}$.

\begin{lemma}\label{lemma:prerigid1}
Let $\mathcal{S} = (\rho, \{ \mathcal{E}_i^A \}, \{ \mathcal{E}_i^B \})$ be Alice and Bob's $\epsilon$-structured strategy for playing $N$ sequential $K$-round steering games in which Alice plays honestly. Let $\mathcal{S}_{g} = (\rho, \{ \mathcal{E}_i^A \}, \{ \mathcal{G}_i^B \})$ be Alice and Bob's $\epsilon$-structured strategy for playing $N$ sequential $K$-round steering games, in which Alice plays as in $\mathcal{S}$ but also guesses Bob's outcomes. Specifically, Bob's operator will be $\mathcal{G}_i^B$, which yields Alice's guesses for Bob's outcomes. We have that $\mathcal{S}$ $O(N\epsilon^{1/6})$-simulates $\mathcal{S}_{g}$.
\end{lemma}
\begin{proof}
Without loss of generality, it can be assume that the initial state $\rho$ is a pure state.
Also note that Alice's action is given by $ \{ \mathcal{E}_i^A \}$ in both strategies, since she is always honest.
We have that $\mathcal{S}_{g} = (\rho, \{ \mathcal{E}_i^A \}, \{ \mathcal{G}_i^B \})$, where $\mathcal{G}_i^B$ denotes the guessing operator for Bob's outcome in game $i$. 
Because Alice is trusted and playing honestly, her guessing strategy for Bob will be to provide the same outcomes as her measurement outcomes. This, of course, ensures that the steering correlations are saturated.
Concretely, $\mathcal{G}_i^B$ leaves Bob's system unchanged but gives the same outcome as $\mathcal{E}_i^A$. Denote as $\mathcal{G}_i^{AB} = \mathcal{E}_i^{A} \circ \mathcal{G}_i^{B}$ the action of Alice and Bob on the state of the $i$'th game, in strategy $\mathcal{S}_{g}$ and similarly
$\mathcal{E}_i^{AB} = \mathcal{E}_i^{A} \circ \mathcal{E}_i^{B}$ the action of Alice and Bob on the state of the $i$'th game, corresponding to the true strategy $\mathcal{S}$.
Note that we can assume the same Hilbert space in both strategies. The reason for this is that Alice's side is trusted in the two strategies, so her Hilbert space is determined and fixed. On Bob's side, assume we have Hilbert space $\mathcal{H}_B$ in strategy $\mathcal{S}$. In strategy $\mathcal{S}_{g}$ we are ignoring Bob's outcomes and replacing them with Alice's guesses. Therefore, we can assume any Hilbert space on Bob's side, so without loss of generality we assume it is $\mathcal{H}_B$. Thus strategies $\mathcal{S}$ and $\mathcal{S}_{g}$ use the same Hilbert space. \\
We denote $\mathcal{G}_{1,i}^{AB} = \mathcal{G}_1^{AB} \circ \mathcal{G}_2^{AB} \circ ... \mathcal{G}_i^{AB} $ and similarly
$\mathcal{E}_{1,i}^{AB} = \mathcal{E}_1^{AB} \circ \mathcal{E}_2^{AB} \circ ... \mathcal{E}_i^{AB} $.
Additionally, according to Definition~\ref{def:sim}, $R$ is a set of random indices, each index taken from a different steering game. We can write $R = \{ \kappa_i | i \in \{1 ... N \}, \kappa_i = (i - 1)K + r, \; r \in_R \{1 ... K \} \}$. Essentially, $\kappa_i$ is the selected random round for the $i$th steering game, out of the total number of rounds in the $N$ games. Thus, $1 \leq \kappa_i \leq NK$.
Lastly, we denote $R_j = \{ \kappa_i | 1 \leq i \leq j \}$ to be the first $j$ random indices (i.e. the randomly selected rounds up to game $j$). 
We would like to compute a bound for:
\begin{equation*}
TD(Tr_{-R_j}(\mathcal{E}_{1,j}^{AB}(\rho)), Tr_{-R_j}(\mathcal{G}_{1,j}^{AB}(\rho)))
\end{equation*}
We do this inductively, starting with the first steering game. Theorem~\ref{thm:azuma} tells us that for a $K$-round steering game, if we observe a steering inequality saturation of order $O(\epsilon)$ the reduced state in a randomly chosen round will be $O(\epsilon^{1/6})$ close to a Bell state. Since we know that strategies $\mathcal{S}$ and $\mathcal{S}_g$ are $\epsilon$-structured, it follows that for the first steering game we have that:
\begin{equation*}
TD(Tr_{-R_1}(\Phi(\mathcal{E}_{1}^{AB}(\rho))), \hat{\mathcal{E}}^{AB}_{1}(\Ket{\psi_+}\Bra{\psi_+})) \leq O(\epsilon^{1/6})
\end{equation*}
\begin{equation*}
TD(Tr_{-R_1}(\Phi(\mathcal{G}_{1}^{AB}(\rho))), \hat{\mathcal{G}}^{AB}_{1}(\Ket{\psi_+}\Bra{\psi_+})) \leq O(\epsilon^{1/6})
\end{equation*}
Where $\Phi$ is the isometry of Theorem~\ref{theorem:selftesting}, $\hat{\mathcal{E}}^{AB}_{1}$ and $\hat{\mathcal{G}}^{AB}_{1}$ are the ideal operators (note that in fact $\hat{\mathcal{G}}^{AB}_{1} = \mathcal{G}^{AB}_{1}$ since Alice is honest). However, since $I \otimes M\Ket{\psi_+} = (XM^TX) \otimes I \Ket{\psi_+}$, for any operator $M$, in $\hat{\mathcal{E}}^{AB}_{1}$ we can ``shift'' Bob's action to Alice's side and so we are left with an operator which is ideal on Alice's side (conjugation by $X$ is not a problem, since $XXX=X$ and $XYX=-Y$) and acts as identity on Bob's side. This is precisely the $\hat{\mathcal{G}}^{AB}_{1}$ operator, hence $\hat{\mathcal{E}}^{AB}_{1} = \hat{\mathcal{G}}^{AB}_{1}$. Therefore, both $Tr_{-R_1}(\Phi(\mathcal{E}_{1}^{AB}(\rho)))$ and $Tr_{-R_1}(\Phi(\mathcal{G}_{1}^{AB}(\rho)))$ are close to the same state, and through a simple application of the triangle inequality it follows:
\begin{equation*}
TD(Tr_{-R_1}(\Phi(\mathcal{E}_{1}^{AB}(\rho))), Tr_{-R_1}(\Phi(\mathcal{G}_{1}^{AB}(\rho)))) \leq O(\epsilon^{1/6})
\end{equation*}
The specific isometry we have considered acts locally in each round (i.e., in each round it only acts on the reduced state of that round and introduces one ancilla qubit). Thus, it is possible in this case to commute the isometry with the tracing out operation, and then remove it completely from the inequality, since it is distance preserving. 
Of course, the tracing out operation will be acting on a different system (the non-isometrized system), requiring a different notation. In order to avoid complicating the notation, we use the same expression for the tracing out operation, and it is to be understood that we trace out all systems apart from the ones used in round $1$. We therefore have:
\begin{equation*}
TD(Tr_{-R_1}(\mathcal{E}_{1}^{AB}(\rho)), Tr_{-R_1}(\mathcal{G}_{1}^{AB}(\rho))) \leq O(\epsilon^{1/6})
\end{equation*}
This is the base case of our induction. Assume now the following holds:
\begin{equation} \label{eq:ind}
TD( Tr_{-R_{j-1}}(\mathcal{E}_{1,j-1}^{AB}(\rho)), Tr_{-R_{j-1}}(\mathcal{G}_{1,j-1}^{AB}(\rho))) \leq  (j-1)O(\epsilon^{1/6})
\end{equation}
We would like to show:
\begin{equation*}
TD(Tr_{-R_{j}}(\mathcal{E}_{1,j}^{AB}(\rho)), Tr_{-R_{j}}(\mathcal{G}_{1,j}^{AB}(\rho))) \leq  jO(\epsilon^{1/6})
\end{equation*}
For the set of rounds from game $j$ we can again apply Theorem~\ref{thm:azuma}, and use the closeness bound for the state from round $\kappa_j$. Note that for the first game, the shared state of Alice and Bob was $\Ket{\psi}$, while the state in game $j$ need not be a pure state. We can, however, still apply Theorem~\ref{thm:azuma}, regardless of the action of previous games, since the theorem assumes that Alice and Bob share some state $\sigma$ which can be either pure or mixed.
Moreover, as before, we will have that $\hat{\mathcal{E}}^{AB}_{1,j} = \hat{\mathcal{G}}^{AB}_{1,j}$, hence the reduced state in round $\kappa_j$ from strategy $\mathcal{S}$ is $O(\epsilon^{1/6})$-close to the state in round $\kappa_j$ from strategy $\mathcal{S}_g$. This together with Equation~\ref{eq:ind} and a triangle inequality lead to:
\begin{equation*}
TD(Tr_{-R_{j}}(\mathcal{E}_{1,j}^{AB}(\rho)), Tr_{-R_{j}}(\mathcal{G}_{1,j}^{AB}(\rho))) \leq  jO(\epsilon^{1/6})
\end{equation*}
Since this is true for all $j \leq N$, we get:
\begin{equation*}
TD(Tr_{-R}(\mathcal{E}_{1,N}^{AB}(\rho)), Tr_{-R}(\mathcal{G}_{1,N}^{AB}(\rho))) \leq  O(N\epsilon^{1/6})
\end{equation*}
Hence, strategy $\mathcal{S}$ $O(N \epsilon^{1/6})$-simulates strategy $\mathcal{S}_g$.
\end{proof}

\begin{lemma}\label{lemma:prerigid2}
Let $\mathcal{S}_{g} = (\rho, \{ \mathcal{E}_i^A \}, \{ \mathcal{G}_i^B \})$ be Alice and Bob's $\epsilon$-structured strategy for playing $N$ sequential $K$-round steering games, in which Alice plays as in $\mathcal{S}$ but also guesses Bob's outcomes. Specifically, Bob's operator will be $\mathcal{G}_i^B$, which yields Alice's guesses for Bob's outcomes. 
Let $\mathcal{S}_{id} = (\rho_{id}, \{ \mathcal{E}_{i}^A \}, \{ \mathcal{E}_{id \; i}^B \})$ be the ideal strategy in which
$\rho_{id}$ is a tensor product of Bell pairs and Alice and Bob play $N$ sequential $K$-round steering games ideally (i.e. they measure the same operators on their shared state). We have that $\mathcal{S}_{g}$ $O(N\epsilon^{1/6})$-simulates an isometric extension of $\mathcal{S}_{id}$, or $\mathcal{S}_{g} \approx \mathcal{S}_{id}$.
\end{lemma}
\begin{proof}
We shall first consider another guessing strategy $\hat{\mathcal{S}}_{g} = (\rho_{id}, \{ \mathcal{E}_i^A \}, \{ \mathcal{G}_i^B \})$. This strategy is identical to $\mathcal{S}_{g}$ except for the fact that it uses the ideal state (the tensor product of Bell states) as opposed to the real state, $\rho$. 
Note that $\rho_{id}$ must lie in the same Hilbert space as $\rho$. Without loss of generality, we can use some isometry to take a tensor product of Bell states and map it to the Hilbert space of $\rho$ and denote that state as $\rho_{id}$.
It is easy to see that $\mathcal{S}_{g} \approx \hat{\mathcal{S}}_{g}$. On the one hand, both strategies use the same operators for Alice and Bob (and in fact will produce identical statistics). On the other hand, since Alice is effectively guessing for Bob and the action on his subsystem is identity, there is no adaptivity in his strategy or outcomes.
Combining this with the fact that the strategies are $\epsilon$-structured, means that we can extract individual Bell pairs from each steering game and hence:
\begin{equation*}
TD(Tr_{-R}(\mathcal{G}_{1,N}^{AB}(\rho)), Tr_{-R}(\mathcal{G}_{1,N}^{AB}(\rho_{id}))) \leq  O(N\epsilon^{1/6})
\end{equation*}
We now show $\hat{\mathcal{S}}_{g} \approx \mathcal{S}_{id}$.
First note that, as before, $\{ \mathcal{E}_{i}^A \}$ corresponds to the ideal strategy for Alice and appears in both $\hat{\mathcal{S}}_g$ and $\mathcal{S}_{id}$ since she is always honest.
Additionally, as mentioned in the previous proof, because she is playing honestly according to the ideal strategy, her guesses for Bob's outcomes will be exactly her own measurement outcomes. This mimics the ideal strategy. In the ideal strategy Alice and Bob measure the same operators on their shared state. Moreover, both strategies are using the ideal state $\rho_{id}$.
This directly implies that:
\begin{equation*}
TD(Tr_{-R}(\mathcal{G}_{1,N}^{AB}(\rho_{id})), Tr_{-R}(\mathcal{E}_{id \; 1,N}^{AB}(\rho_{id}))) \leq  O(N\epsilon^{1/6})
\end{equation*}
Alternatively, we could have used the fact that $\hat{\mathcal{S}}_{g}$ and $\mathcal{S}_{id}$ use the same state and both strategies are $\epsilon$-structured, as in the previous Lemma, to prove the same thing.
Thus, it follows that strategy $\mathcal{S}_{g}$ $O(N\epsilon^{1/6})$-simulates an isometric extension of strategy $\mathcal{S}_{id}$, or $\mathcal{S}_{g} \approx \mathcal{S}_{id}$.
\end{proof}

It is worth mentioning that in the case of steering, unlike CHSH, the outcomes are deterministic, so Alice can guess perfectly for Bob. This led to a simplified rigidity proof compared to the one for CHSH games from \cite{ruv}.
Finally, we prove Theorem~\ref{theorem:rigidity}.
\begin{proof}
From Lemmas~\ref{lemma:prerigid1} and~\ref{lemma:prerigid2} we have that $\mathcal{S} \approx \mathcal{S}_{g}$ and that
$\mathcal{S}_{g} \approx \mathcal{S}_{id}$. 
Since in Lemma~\ref{lemma:prerigid2} we have shown that $\mathcal{S}_{g}$ $O(N\epsilon^{1/6})$-simulates an isometric extension of $\mathcal{S}_{id}$ we can consider an isometric extension of $\mathcal{S}_{id}$, denoted $\mathcal{S}_{id}'$ that has the same Hilbert space as strategy $\mathcal{S}$.
Using a triangle inequality it follows that $\mathcal{S} \approx \mathcal{S}_{id}'$. Thus $\mathcal{S}$ $O(N\epsilon^{1/6})$-simulates an isometric extension of the ideal strategy $\mathcal{S}_{id}$.
\end{proof}

\section{Proof of Theorem~\ref{theorem:totallysteerable}}\label{proof:Thm5}
It is clear that a maximally entangled Bell state satisfies the properties of total steerability. We therefore focus on proving that a totally steerable state is maximally entangled.
From the two constraints of definition~\ref{def:total} we will express the most general form of $\rho_{AB}$.
We start by considering $\Ket{\psi_{ABC}}$ as the $4$-qubit purification of $\rho_{AB}$. All other purifications are equivalent to this one, so this suffices for our purposes.
Writing $\Ket{\psi_{ABC}}$ in the computational basis, we have:
\begin{equation*}
\Ket{\psi_{ABC}} = \sum \limits_{i = 0}^{15} a_i \Ket{i}
\end{equation*}
Of course, we have the additional constraint:
\begin{equation*}
\sum \limits_{i = 0}^{15} |a_i|^2 = 1
\end{equation*}
By re-expressing the constraints from definition~\ref{def:total} and lemma~\ref{lemma:complete}:
\begin{equation*}
\rho_{BC} = \rho_B \otimes \rho_C
\end{equation*}
\begin{equation*}
\rho_B = I/2
\end{equation*}
as constraints on the amplitudes of $\Ket{\psi_{ABC}}$ we will have a large bilinear system of equations. From this system we will arrive at the following set of equations:
\begin{eqnarray*}
a_{2k} &=& f \cdot a_{2k + 2}, \; \; k \in \{0, 1, ... 6 \} \\
a_{2k + 3} &=& -f^* \cdot a_{2k + 1}, \; \; k \in \{0, 1, ... 6 \} \\
a_{4k + 1} &=& e^{i\phi_1} \cdot a_{4k + 2}, \; \; k \in \{0, 1, 2, 3 \} \\
a_{4k} &=& e^{i\phi_2} \cdot a_{4k + 3}, \; \; k \in \{0, 1, 2, 3 \}
\end{eqnarray*}

Where the parameters we introduced are $f \in \mathbb{C}$ and $\phi_1, \phi_2 \in [0, 2\pi]$.
Computing the matrix elements of $\rho_{AB}$, we arrive at the most general form, given by:
\begin{equation*}
\rho_{AB} = \frac{1}{2(|f|^2 + 1)} \left( \begin{array}{cccc}
|f|^2 & f & f & e^{i \phi_1} |f|^2 \\
f^* & 1 & e^{i \phi_2} & -f \\
f^* & e^{-i \phi_2} & 1 & -f \\
e^{-i \phi_1} |f|^2 & -f^* & -f^* & |f|^2 \\
\end{array} \right)
\end{equation*}
It can be easily checked that $Tr(\rho_{AB}^2) = 1$ and therefore $\rho_{AB}$ is a pure state. But since $\rho_B = I/2$, we have that $\rho_{AB}$ is a pure entangled state.
We know that all pure entangled states are non-local. Moreover, the condition $\rho_B = I/2$ also implies that the state is maximally non-local and hence a Bell state. 

\bibliography{report}
\bibliographystyle{unsrt}

\end{document}